\pdfoutput=1
\documentclass{lmcs} 
\pdfoutput=1

\usepackage[utf8]{inputenc}

\usepackage{lastpage}
\lmcsdoi{19}{1}{2}
\lmcsheading{}{\pageref{LastPage}}{}{}%
{Mar.~11,~2022}{Jan.~12,~2023}{}


\usepackage{hyperref,amsmath,amssymb}
\usepackage{stmaryrd}
\usepackage[arrow, matrix, tips, curve, graph, rotate]{xy}
\SelectTips{cm}{10}

\newcommand{\C}{\mathcal{C}}
\newcommand{\E}{\mathcal{E}}
\newcommand{\thg}{{\mathord{\text{--}}}}
\newcommand{\dbr}[1]{\left\llbracket{#1}\right\rrbracket}
\newcommand{\abs}[1]{{\left|{#1}\right|}}
\newcommand{\alg}[1][X]{{\boldsymbol{#1}}}
\newcommand{\quot}{\delimiter"502F30E\mathopen{}}
\newcommand{\cat}[1]{\mathit{#1}}
\newcommand{\cd}[2][]{\vcenter{\hbox{\xymatrix#1{#2}}}}

\begin{document}

\title{Stream processors and comodels}

\author[R.~Garner]{Richard Garner\lmcsorcid{0000-0003-4475-8721}}
\address{School of Mathematical and Physical Sciences,  Macquarie University, NSW 2109, Australia}
\email{richard.garner@mq.edu.au}
\thanks{Supported by ARC grants FT160100393 and DP190102432.}

\begin{abstract}
  In 2009, Hancock, Pattinson and Ghani gave a coalgebraic
  characterisation of stream processors
  $A^\mathbb{N} \rightarrow B^\mathbb{N}$ drawing on ideas of
  Brouwerian constructivism. Their stream processors have an
  \emph{intensional} character; in this paper, we give a corresponding
  coalgebraic characterisation of \emph{extensional} stream
  processors, i.e., the set of continuous functions
  $A^\mathbb{N} \rightarrow B^\mathbb{N}$. Our account sites both our
  result and that of \emph{op.\,cit.} within the apparatus of
  \emph{comodels} for algebraic effects originating with
  Power--Shkaravska. Within this apparatus, the distinction between
  intensional and extensional equivalence for stream processors arises
  in the same way as the the distinction between \emph{bisimulation}
  and \emph{trace equivalence} for labelled transition systems and
  probabilistic generative systems.
\end{abstract} 

\maketitle

\section{Introduction}
\label{sec:introduction}

As is well known, the type of infinite \emph{streams} of elements of
some type $A$ may be defined to be the final coalgebra
$\nu X.\, A \times X$. If types are mere sets, then this coalgebra is
manifested as the set $A^\mathbb{N}$ of infinite lists of
$A$-elements, with the structure map
\begin{equation}
   \qquad \label{eq:3}
   \alpha \colon \vec a  \mapsto (a_0, \partial \vec a) \qquad \text{where} \qquad \partial(a_0, a_1, a_2, \dots, ) = (a_1, a_2,
   \dots)\rlap{ .}
\end{equation}
Of course, the coalgebra structure describes the corecursive nature of
streams, but also captures their sequentiality: an $A$-stream is
\emph{first} an $A$-value, and \emph{then} an $A$-stream.

If $A$ and $B$ are types, then an \emph{$A$-$B$-stream processor} is a
way of turning an $A$-stream into a $B$-stream. If types are sets,
then the crudest kind of stream processor would simply be a function
$f \colon A^\mathbb{N} \rightarrow B^\mathbb{N}$; however, it is more
computationally reasonable to restrict to those $f$ which are
\emph{productive}, in the sense that determining each $B$-token of the
output should require examining only a finite number of $A$-tokens of
the input.

The productive functions
$f \colon A^\mathbb{N} \rightarrow B^\mathbb{N}$ are in fact precisely
the \emph{continuous} ones for the prodiscrete (= Baire) topologies on
$A^\mathbb{N}$ and $B^\mathbb{N}$. While this representation of stream
processors is mathematically smooth, it fails to make explicit their
sequentiality: we should like to see the fact that determining each
\emph{successive} token of the output $B$-stream requires examining
\emph{successive} finite segments of the input $A$-stream. Much as for
streams themselves, this can be done by presenting stream processors
as a final coalgebra.

Such a presentation was given in~\cite{Hancock2009Representations}.
Therein, the \emph{type of $A$-$B$-stream processors} was taken to be
the final coalgebra $\nu X.\, T_A(B \times X)$, where
$T_A(V) = \mu X.\, V + X^A$; and it was explained how each element of
this type encodes a continuous function
$A^\mathbb{N} \rightarrow B^\mathbb{N}$, and how, conversely, each
such continuous function yields an element of this type. An
interesting aspect of the story is that these assignments are
\emph{not} mutually inverse: distinct elements of
$\nu X.\, T_A(B \times X)$ may represent the same continuous function,
so that elements of this type are really \emph{intensional}
representations of stream-processing algorithms.

While there are many perspectives from which this is a good thing, it
leaves open the question of whether there is a coalgebraic
representation for \emph{extensional} stream processors, i.e.,
for the set of continuous functions
$A^\mathbb{N} \rightarrow B^\mathbb{N}$. In this paper, we
show that there is:

\begin{thm}
  \label{thm:1}
  Let $A$ and $B$ be sets. The set of continuous functions
  $A^\mathbb{N} \rightarrow B^\mathbb{N}$ is the underlying set of the
  terminal $B$-ary comagma in the category of $A$-ary magmas.
\end{thm}

In this result, an \emph{$A$-ary magma} is a set $X$ with an operation
$\xi \colon X^A \rightarrow X$ satisfying no further axioms. More
generally, we can speak of $A$-ary magmas in any category $\C$ with
products; while, if $\C$ is a category with \emph{co}products, we can
define an $A$-ary \emph{co}magma in $\C$ to be an $A$-ary magma in
$\C^\mathrm{op}$. Explicitly, this involves an object $X \in \C$ and a
map $X \rightarrow X + \dots + X$ into the coproduct of $A$ copies of
$X$, subject to no further conditions.

On the face of it, our Theorem~\ref{thm:1} has no obvious relation
to~\cite{Hancock2009Representations}, nor to anything resembling
computation. Thus, the broader contribution of this paper is to site
both the ideas of~\cite{Hancock2009Representations} and our
Theorem~\ref{thm:1} within the 
well-established machinery of
\emph{comodels}~\cite{Power2004From, Plotkin2008Tensors}, as we now
explain.

The category-theoretic approach to computational effects originates
in~\cite{Moggi1991Notions}: given a monad $\mathsf{T}$ on a category
of types and programs, we view elements of $T(V)$ as computations with
side-effects from $\mathsf{T}$ returning values in $V$. This idea was
refined in~\cite{Plotkin2002Notions}; rather than considering monads
\emph{simpliciter}, we generate them from \emph{algebraic theories}
whose basic operations are the computational primitives for the
effects at issue. A key example, for us, is the theory $\mathbb{T}_A$
of input from an alphabet $A$, which is freely generated by a single
$A$-ary operation $\mathsf{read}$.

The approach via algebraic theories has the virtue of giving a good
notion of \emph{model} in any category with finite powers. In
particular, one has \emph{comodels}, which are models in the opposite
of the category of sets, and a key insight of~\cite{Power2004From} is
that comodels of a theory $\mathbb{T}$ can be seen as coalgebraic
objects for evaluating $\mathbb{T}$-computations to values. We recall
these developments in detail in Section~\ref{sec:streams-as-final},
and in particular, we see that a $\mathbb{T}_A$-comodel is an
$A \times (\thg)$-coalgebra, and that the \emph{final} comodel is the
set of $A$-streams.

A range of authors~\cite{Plotkin2008Tensors, Mogelberg2014Linear,
  Pattinson2015Sound, Pattinson2016Program, Uustalu2015Stateful,
  Katsumata2020Interaction, Ahman2020Runners, Goncharov2020Toward,
  Uustalu2020Algebraic, Garner2020The-costructure-cosemantics} have
taken this attractive perspective on operational semantics further.
Particularly salient for us is the concept, due
to~\cite{Ahman2020Runners, Katsumata2020Interaction,
  Uustalu2020Algebraic} of a \emph{residual comodel}. Given theories
$\mathbb{T}$ and $\mathbb{R}$, an $\mathbb{R}$-residual
$\mathbb{T}$-comodel is, formally speaking, a comodel of $\mathbb{T}$
in the Kleisli category of $\mathbb{R}$ but is, practically speaking,
a coalgebraic entity for evaluating or compiling
$\mathbb{T}$-computations into $\mathbb{R}$-computations. In
particular, we have $\mathbb{T}_A$-residual $\mathbb{T}_B$-comodels,
which translate requests for $B$-input into requests for $A$-input,
and a little thought shows that this is exactly the r\^ ole filled by
an $A$-$B$-stream processor. In fact, the final coalgebra
of~\cite{Hancock2009Representations} turns out to be precisely the
\emph{final $\mathbb{T}_A$-residual $\mathbb{T}_B$-comodel}; in
Sections~\ref{sec:runn-resid-runn} and~\ref{sec:intens-stre-proc}, we
explain this, and show how other aspects
of~\cite{Hancock2009Representations} such as the \emph{composition} of
intensional stream processors flow naturally.

\looseness=-1 To get from here to our Theorem~\ref{thm:1} requires a
new import from category-theoretic universal algebra: the notion of a
\emph{bimodel}~\cite{Freyd1966Algebra, Tall1970Representable,
  Bergman1996Co-groups}. An \emph{$\mathbb{R}$-$\mathbb{T}$-bimodel}
is a comodel of $\mathbb{T}$ in the category of $\cat{Set}$-models of
$\mathbb{R}$. Since this latter category has the Kleisli category as a
full subcategory, bimodels are a generalisation of residual
comodels---one which, roughly speaking, allows additional quotients to
be taken. These quotients are just what one needs to collapse the
intensional stream processors of~\cite{Hancock2009Representations} to
their underlying continuous functions. We develop this theory in
Section~\ref{sec:extens-stre-proc}, culminating in our
Theorem~\ref{thm:1} which we now recognise as describing the
\emph{final $\mathbb{T}_A$-$\mathbb{T}_B$-bimodel}.

An obvious question at this point is whether we have similar
characterisations of the final bimodel on replacing $\mathbb{T}_A$ and
$\mathbb{T}_B$ by more elaborate algebraic theories. One step in this
direction is given in~\cite{Yoshida2022Continuous}, which
characterises the final $\mathbb{R}$-$\mathbb{T}$-bimodel whenever
$\mathbb{R}$ and $\mathbb{T}$ are \emph{free} algebraic theories,
i.e., theories generated by operations subject to no equations.

A different direction of generalisation points towards labelled
transition systems and generative probabilistic
systems~\cite{Glabbeek1990Reactive}. Indeed, (non-terminating,
finitely branching) labelled transition systems with alphabet $A$ are
precisely $\smash{\mathbb{P}_f^+}$-residual $\mathbb{T}_A$-comodels,
for $\smash{\mathbb{P}_f^+}$ the theory of binary non-deterministic
choice; while (finitely supported) generative probabilistic systems
with alphabet $A$ are $\mathbb{D}$-residual $\mathbb{T}_A$-comodels,
for $\mathbb{D}$ the theory of binary probabilistic choice. As is well
known, in these examples, the final residual $\mathbb{T}_A$-comodel
captures states up to \emph{bisimulation} equivalence. What is perhaps
less well known is that the final \emph{bimodel} captures states up to
\emph{trace} equivalence; indeed, as shown in~\cite[\S
7]{Garner2018Hypernormalisation}, the final
$\smash{\mathbb{P}_f^+}$-$\mathbb{T}_A$-bimodel is the set of closed
subsets of $A^\mathbb{N}$; while the final
$\mathbb{D}$-$\mathbb{T}_A$-bimodel is the set of probability measures
on $A^\mathbb{N}$. This fact provides an alternative perspective on
the trace semantics of~\cite{Hasuo2007Generic} (which itself builds
on~\cite{Power1999Coalgebraic}) in which an object of traces is found
as a final object among not all bimodels, but among all \emph{free}
bimodels; in future work, we will give a more careful comparison of
the two notions of trace.


From this perspective, then, the continuous function encoded by an
intensional stream processor can be seen as its ``trace'', and with
this in mind, the final contribution of this paper in
Section~\ref{sec:comp-intens-extens} is to explain from a
comodel-theoretic perspective the fact that intensional stream
processors admit a procedure of ``normalisation-by-trace-evaluation'',
which normalises each intensional stream processor to a
\emph{maximally lazy} stream processor with the same trace; this is a
particular instantiation of a more general schema which is explored
further in~\cite{Garner2018Hypernormalisation}.

\section{Streams as a final comodel}
\label{sec:streams-as-final}

In this background section, we recall how algebraic theories present
notions of effectful computation, how \emph{comodels} of a theory
furnish environments appropriate for evaluating such computations, and
how the type of streams arises as a final comodel.

\begin{defi}[Algebraic theory]
  \label{def:1}
  A \emph{signature} comprises a set $\Sigma$ of \emph{function
    symbols}, and for each $\sigma \in \Sigma$ a set $\abs \sigma$,
  its \emph{arity}. Given a signature $\Sigma$ and a set $V$, we
  define the set $\Sigma(V)$ of \emph{$\Sigma$-terms with variables in
    $V$} by the inductive clauses
  \begin{equation*}
    v \in V \implies v \in \Sigma(V)
    \quad \text{and} \quad
    \sigma \in \Sigma\text{, }t \in \Sigma(V)^{\abs \sigma} \implies
    \sigma(t)
    \in \Sigma(V)\rlap{ .}
  \end{equation*}
  An \emph{equation} over a signature $\Sigma$ is a formal equality $t=u$ between
  terms in the same set of free variables. A (algebraic) \emph{theory}
  $\mathbb{T}$ comprises a signature and a
  set $\E$ of equations over it.
\end{defi}

\begin{defi}[$\mathbb{T}$-terms]
  \label{def:8}
  Given a signature $\Sigma$ and terms $t \in \Sigma(V)$ and
  $u \in \Sigma(W)^V$, we define the \emph{substitution}
  $t(u) \in \Sigma(W)$ by recursion on $t$:
  \begin{equation}\label{eq:20}
    \quad v \in V \implies v(u) = u_v \quad \text{and} \quad 
    \sigma \in
    \Sigma\text{, } t \in \Sigma(V)^{\abs \sigma} \implies
    (\sigma(t))(u) = \sigma(\lambda i.\,t_i(u))\rlap{ .}
  \end{equation}
  Given a theory $\mathbb{T}$ with signature $\Sigma$, we define
  \emph{$\mathbb{T}$-equivalence} as the smallest family of
  substitution-congruences $\equiv_\mathbb{T}$ on the sets $\Sigma(V)$
  such that $t\equiv_\mathbb{T} u$ for all equations $t=u$ of $\mathbb{T}$.
  The set $T(V)$ of \emph{$\mathbb{T}$-terms with variables in $V$} is
  $\Sigma(V) \quot \equiv_\mathbb{T}$.
\end{defi}

When a theory $\mathbb{T}$ is seen as specifying a computational
effect, $T(V)$ describes the set of computations with effects from
$\mathbb{T}$ returning a value in $V$. 

\begin{exa}[Non-deterministic choice]
  The theory $\mathbb{P}_f^+$ of \emph{non-deterministic choice}
  comprises a single binary function symbol $\vee$ (written in infix
  notation) together with the equations
  \begin{equation*}
    x \vee y = y \vee x \qquad \qquad x \vee x = x \qquad \qquad (x \vee y) \vee z = x \vee (y \vee z)\rlap{ .}
  \end{equation*}
  The set of terms $P_f^+(V)$ can be identified with the set of
  non-empty finite subsets of $V$, where the subset $\{v_1, \dots,
  v_n\}$ corresponds to the term $v_1 \vee \cdots \vee v_n$. We view
  this term as encoding a program which chooses non-deterministically
  between one of the return values $v_1, \dots, v_n$.
\end{exa}

\begin{exa}[Probabilistic choice]
  The theory $\mathbb{P}_f^+$ of \emph{probabilistic choice}
  comprises a family of binary function symbols $+_r$ indexed by $r
  \in (0,1)$ together with the equations
  \begin{equation*}
    x +_r y = y +_{r^\ast} x \qquad \qquad x +_r x = x \qquad \qquad (x +_r y) +_s z = x +_{rs} (y +_{r^\ast s/(rs)^\ast} z)
  \end{equation*}
  where we write $r^\ast$ for $1-r$. The set of terms $D(V)$ can be
  identified with the set of finitely supported discrete probability
  distributions on $V$; we see this as a program which chooses
  probabilistically among possible return values in $V$.
\end{exa}

\begin{exa}[Input]
  \label{ex:6} Given a set $A$, the theory $\mathbb{T}_A$ of
  \emph{$A$-valued input} comprises a single $A$-ary function symbol
  $\mathsf{read}$, satisfying no equations, whose action we think of
  as:
  \begin{equation*}
    (t \colon A \rightarrow X) \mapsto 
    \text{\textsf{let read$()$ be $a$.\,$t(a)$}}\rlap{ .}
  \end{equation*}
  The set of terms $T_A(V)$ is, as in the introduction, the initial
  algebra $\mu X.\, V + X^A$, whose elements may be seen
  combinatorially as $A$-ary branching trees with leaves labelled in
  $V$; or computationally as programs which request $A$-values from an
  external source and use them to determine a return value in $V$. For
  example, when $A = \mathbb{N}$, the program which requests two input
  values and returns their sum is presented by
  \begin{equation}\label{eq:15}
    \text{\textsf{let read$()$ be $n$.\,let read$()$ be $m$.\,$n+m$}}
    \ \in \ T(\mathbb{N})\rlap{ .}
  \end{equation}
\end{exa}

We now define the models of an algebraic theory. In the definition, we
say that a category $\C$ has \emph{powers} if it has all set-indexed
self-products $X^A := \Pi_{a \in A} X$.

\begin{defi}[$\Sigma$-structure, $\mathbb{T}$-model]
  \label{def:17}
  Let $\Sigma$ be a signature. A \emph{$\Sigma$-structure $\alg$} in a
  category $\C$ with powers is an object $X \in \C$ with
  \emph{operations}
  $\dbr {\sigma}_{\alg} \colon X^{\abs \sigma} \rightarrow X$ for each
  $\sigma \in \Sigma$. For each $t \in \Sigma(V)$ the
  \emph{derived operation} $\dbr{t}_{\alg} \colon X^V \rightarrow X$
  is then determined by the recursive clauses:
  \begin{equation}\label{eq:24}
    \dbr{v}_{\alg} = \pi_v \qquad \text{and} \qquad
    \smash{\dbr{\sigma(t)}_{\alg} = X^V
      \xrightarrow{(\dbr{t_i}_{\alg})_{i \in \abs \sigma}} X^{\abs
        \sigma} \xrightarrow{\dbr{\sigma}_{\alg}} X}\rlap{ .}
  \end{equation}
  Given a theory $\mathbb{T} = (\Sigma, \E)$, a
  \emph{$\mathbb{T}$-model in $\C$} is a $\Sigma$-structure $\alg$
  which satisfies $\dbr{t}_{\alg} = \dbr{u}_{\alg}$ for all equations
  $t=u$ of $\mathbb{T}$. The unqualified term ``model''
  will mean ``model in $\cat{Set}$''.

  A \emph{homomorphism} $f \colon \alg \rightarrow \alg[Y]$ of
  $\mathbb{T}$-models in $\C$ is a $\C$-map $f \colon X \rightarrow Y$
  such that for all $\sigma \in \Sigma$ we have
  $\dbr{\sigma}_{\alg[Y]} \circ f^{\abs{\sigma}} = f \circ
  \dbr{\sigma}_{\alg}$. We write $\cat{Mod}(\mathbb{T}, \C)$ for the
  category of $\mathbb{T}$-models in $\C$, and $\cat{Mod}(\mathbb{T})$
  for the models in $\cat{Set}$.
\end{defi}

The set of computations $T(V)$ has a structure of $\mathbb{T}$-model
$\alg[T](V)$ with operations given by substitution; and as is well
known, this structure is in fact \emph{free}:
\begin{lem}
  \label{lem:4}
  The inclusion of variables $\eta_V \colon V \rightarrow T(V)$
  exhibits $\alg[T](V)$ as the free $\mathbb{T}$-model on $V$. That
  is, for any $\mathbb{T}$-model $\alg$ and any function
  $f \colon V \rightarrow X$, there is a unique $\mathbb{T}$-model
  homomorphism $f^\dagger \colon \alg[T](V) \rightarrow \alg$ with
  $f^\dagger \circ \eta_V = f$. Explicitly,
  $f^\dagger(t) = \dbr{t}_{\alg}\!(\lambda v.\, f(v))$.
\end{lem}

Taking the full subcategory of $\cat{Mod}(\mathbb{T})$ on the free
$\mathbb{T}$-models yields the well known \emph{Kleisli category} of
$\mathbb{T}$, which we typically present as follows:
\begin{defi}[Kleisli category]\label{def:7}
  The \emph{Kleisli category} $\cat{Kl}(\mathbb{T})$ of a theory
  $\mathbb{T}$ has sets as objects; hom-sets
  $\cat{Kl}(\mathbb{T})(A,B) = \cat{Set}(A,TB)$; the identity at $A$
  being $\eta_A \colon A \rightarrow TA$; and composition 
  $g,f \mapsto g^\dagger \circ f$ with $g^\dagger$ as in
  Lemma~\ref{lem:4} for the free $\mathbb{T}$-model structures. The
  \emph{free functor}
  $F_\mathbb{T} \colon \cat{Set} \rightarrow \cat{Kl}(\mathbb{T})$ is
  the identity on objects and sends $f \in \cat{Set}(X,Y)$ to
  $\eta_Y \circ f \in \cat{Kl}(\mathbb{T})(X,Y)$. The fully faithful \emph{comparison
  functor} $I_\mathbb{T} \colon \cat{Kl}(\mathbb{T}) \rightarrow
  \cat{Mod}(\mathbb{T})$ maps $A \mapsto \alg[T]A$ and $f \mapsto f^\dagger$.
\end{defi}

The Kleisli category captures the compositionality of computations
with effects from $\mathbb{T}$, and allows us to draw the link with
Moggi's monadic semantics~\cite{Moggi1991Notions}; indeed, the free
functor
$F_\mathbb{T} \colon \cat{Set} \rightarrow \cat{Kl}(\mathbb{T})$ and
its right adjoint $\cat{Kl}(\mathbb{T})(1, \thg) \colon
\cat{Kl}(\mathbb{T}) \rightarrow \cat{Set}$ generate an associated monad
$\mathsf{T}$ on $\cat{Set}$ and we have that
$\cat{Kl}(\mathbb{T}) \cong \cat{Kl}(\mathsf{T})$ under $\cat{Set}$.

So far we have said nothing about \emph{non}-free $\mathbb{T}$-models.
It is a basic fact that every such model can be obtained from a free
one by quotienting by some congruence, and so can been seen as a set of
computations identified up to some notion of program equivalence. This
is important, for example, in~\cite{Levy2003Call-by-push-value}, and
will be important for us in \S\ref{sec:extens-stre-proc} below.

We now turn from models to the dual notion of \emph{comodel}. We say a
category $\C$ has \emph{copowers} if if each set-indexed
self-coproduct $A \cdot X = \Sigma_{a \in A} X$ exists in $\C$.

\begin{defi}[$\mathbb{T}$-comodel]
  \label{def:3}\looseness=-1
  Let $\mathbb{T}$ be a theory. A \emph{$\mathbb{T}$-comodel} in a
  category $\C$ with copowers is a model of $\mathbb{T}$ in
  $\C^\mathrm{op}$, comprising an object $S \in \C$ and
  \emph{co-operations}
  $\dbr{\sigma}^{\alg[S]} \colon S \rightarrow \abs \sigma \cdot S$ 
  obeying the equations of $\mathbb{T}$. The
  unqualified term ``comodel'' will mean ``comodel in
  $\cat{Set}$''. We write
  $\cat{Comod}(\mathbb{T}, \C)$ for the category of $\mathbb{T}$-comodels
  in $\C$, and $\cat{Comod}(\mathbb{T})$ for the comodels in $\cat{Set}$.
\end{defi}

As explained in~\cite{Power2004From, Plotkin2008Tensors}, when a
theory $\mathbb{T}$ presents a notion of computation, its comodels
provide deterministic environments for evaluating computations with
effects from $\mathbb{T}$.

\begin{exa}
  \label{ex:38} A comodel $\alg[S]$ of the theory of
  $A$-valued input is a state machine that answers requests for
  $A$-characters; it comprises a set of states $S$ and a map
  $\dbr{\mathsf{read}}^{\alg[S]} = (g,n) \colon S \rightarrow A \times
  S$ giving for each $s \in S$ a next character $g(s) \in A$ and a
  next state $n(s) \in S$.
\end{exa}

While the comodels of the preceding example are just
$A \times (\thg)$-coalgebras, the comodel perspective adds something
to this. The general picture is that a $\mathbb{T}$-comodel allows us
to evaluate $\mathbb{T}$-computations $t \in T(V)$ down to values in
$V$ via the derived operations of Definition~\ref{def:17}. Indeed,
given a comodel $\alg[S]$ and a term $t \in T(V)$, we have the derived
co-operation $\dbr{t}^{\alg[S]} \colon S \rightarrow V \times S$
which, unfolding the definition, is given by the clauses:
\begin{equation}\label{eq:19}
\begin{aligned}
  v \in V &\implies \dbr{v}^{\alg[S]} (s) = (v,s)  \qquad \ \ \,
  \\ \text{and }\sigma 
  \in \Sigma, t \in T(V)^{\abs \sigma} &\implies
  \dbr{\sigma(t)}^{\alg[S]}(s)= \dbr{t_i}^{\alg[S]}(s') \text{ where }
  \dbr{\sigma}^{\alg[S]}(s) = (i,s')\rlap{ .}
\end{aligned}
\end{equation}
The idea is that applying $\dbr{t}^{\alg[S]}$ to a starting state $s \in S$ will
yield the value $v \in V$ and final state $s' \in S$ obtained by
running the computation $t \in T(V)$, using the co-operations of
the comodel $\alg[S]$ to answer the requests posed by the
corresponding operation symbols of $\mathbb{T}$.

\begin{exa}
  \label{ex:39}
  For a comodel $(g,n) \colon S \rightarrow A \times S$ of $A$-valued input, the clauses~\eqref{eq:19} become
  \begin{equation*}
    v \in V \implies \dbr{v}^{\alg[S]} (s) = (v,s)  \qquad \ \ 
    t \in T(V)^A \implies
    \dbr{\textsf{read}(t)}^{\alg[S]}(s)= \dbr{t(g(s))}^{\alg[S]}(n(s))\rlap{ .}
  \end{equation*}
  So if we consider $A = \mathbb{N}$, the term
  $t = \mathsf{read}(\lambda n.\, \mathsf{read}(\lambda m.\, n+m)) \in
  T(\mathbb{N})$ from~\eqref{eq:15}, and the comodel $\alg[S]$ with
  $S = \{s,s',s''\}$ and
  $\dbr{\mathsf{read}}^{\alg[S]} = (g,n) \colon S \rightarrow \mathbb{N} \times S$ given
  by the upper line in:
  \begin{align*}
    \dbr{\mathsf{read}}^{\alg[S]}: \qquad s & \mapsto (3,s') & 
    s' &\mapsto (6,s'') &
    s'' &\mapsto (11,s'') \\
    \dbr{t}^{\alg[S]}: \qquad   s &\mapsto (9,s'') &
    s' & \mapsto (17,s'') &
    s'' & \mapsto (22,s'')\rlap{ ,}
  \end{align*}
  then $\dbr{t}^{\alg[S]} \colon S \rightarrow \mathbb{N} \times S$ is
  given by the lower line. For example,  we calculate that
  $\dbr{t}(s) = \dbr{\mathsf{read}(\lambda n.\, \mathsf{read}(\lambda m.\, n+m))}(s)
  = \dbr{\mathsf{read}(\lambda m.\, 3+m)}(s') = \dbr{3+6}(s'') =
  (9,s'')$.
\end{exa}

As is idiomatic, the \emph{final} comodel of a theory describes 
``observable behaviours'' that states of a comodel may possess. To
make this precise, we define states $s_1 \in \alg[S]_1$ and
$s_2 \in \alg[S]_2$ of two $\mathbb{T}$-comodels to be
\emph{operationally equivalent} if running any
$\mathbb{T}$-computation $t \in T(V)$ starting from the state $s_1$ of
$\alg[S]_1$ or from
the state $s_2$ of $\alg[S]_2$ gives the same value; i.e., 
\begin{equation*} 
  \text{if} \qquad \dbr{t}^{\alg[S]_1}(s_1) = (v_1,s_1') \quad \text{and} \quad 
  \dbr{t}^{\alg[S]_2}(s_2) = (v_2,s_2') \qquad \text{then} \qquad v_1 = v_2\rlap{ .}
\end{equation*}
\begin{lem}
  States $s_1 \in \alg[S_1]$ and $s_2 \in \alg[S]_2$ of two
  $\mathbb{T}$-comodels are operationally equivalent if and only if
  they become equal under the unique maps
  $\alg[S]_1 \rightarrow \alg[F] \leftarrow \alg[S]_2$ to the final
  $\mathbb{T}$-comodel.
\end{lem}
\begin{proof}
  This is~\cite[Proposition~5.3]{Garner2020The-costructure-cosemantics}.
\end{proof}

So in the spirit of~\cite[Theorem~4]{Kupke2009Characterising}, we may
(if we adequately handle the set-theoretic issues) characterise the
final $\mathbb{T}$-comodel as the set of all possible states of all
possible comodels, modulo operational equivalence. However, a more
algebraic approach is also possible. The following is~\cite[Definition~5.5]{Garner2020The-costructure-cosemantics}:
\begin{defi}[Admissible behaviour]
  An \emph{admissible behaviour} $\beta$ for a theory $\mathbb{T}$ is
  a family of functions $\beta_V \colon TV \rightarrow V$, as $V$
  ranges over sets, such that
  \begin{equation*}
    v \in V \implies \beta_V(v) = v \quad \text{and} \quad
    t \in TV, u \in (TW)^V \implies \beta_W(t(u)) = \beta_W(t \gg u_{\beta_V(t)})\rlap{ ,}
  \end{equation*}
  where for terms $f \in T(A)$
  and $g \in T(B)$, we write $f \gg g$ for the term
  $f(\lambda a.\, g) \in T(B)$.
\end{defi}

An admissible behaviour is a way of evaluating
$\mathbb{T}$-computations to values, and from this perspective, the
two axioms are quite intuitive: for example, the second says that, if
the result of evaluating $t \in TV$ is $v \in V$, then the result of
evaluating $t(u) \in TW$ coincides with that of evaluating the
computation which sequences $t$ (discarding the return value) into
$u_v$.

\begin{exa}\label{ex:admissible}
  Any state $s$ of a $\mathbb{T}$-comodel $\alg[S]$
  yields an admissible $\mathbb{T}$-behaviour $\beta_s$, where for
  $t \in TV$ we define $\beta_s(t)$ to be the first component of
  $\dbr{t}^{\alg[S]}\!(s) \in V \times S$.
\end{exa}

\begin{prop}
  The final $\mathbb{T}$-comodel of a theory $\mathbb{T}$ can be
  described as the set $\alg[F]$ of admissible behaviours, under the
  co-operations
  \begin{equation}\label{eq:11}
    \dbr{\sigma}^{\alg[F]} \colon \beta \mapsto (\beta(\sigma), \partial_\sigma \beta)
  \end{equation}
  where $\partial_\sigma \beta$ is the admissible behaviour given by
  $(\partial_\sigma \beta)(t) = \beta(\sigma \gg t)$. For any
  $\mathbb{T}$-comodel $\alg[S]$, the unique comodel homomorphism
  $\alg[S] \rightarrow \alg[F]$ sends $s$ to $\beta_s$ as in
  Example~\ref{ex:admissible}.
\end{prop}
\begin{proof}
  This is~\cite[Proposition~5.9]{Garner2020The-costructure-cosemantics}.
\end{proof}

For $A$-valued input, the admissible behaviours can be identified (as
in Example~5.10 of \emph{loc.~cit.}) with
streams of $A$-values: given an admissible behaviour $\beta$, the
corresponding stream of values is
\begin{equation*}
  (\beta(\mathsf{read}), \beta(\mathsf{read} \gg \mathsf{read}), \beta(\mathsf{read} \gg \mathsf{read} \gg \mathsf{read}), \dots)\rlap{ .}
\end{equation*}
Under this identification, the co-operation
$\dbr{\mathsf{read}}^{\alg[F]}$ of~\eqref{eq:11} is easily seen to
coincide with the structure map~\eqref{eq:3} on the set of
$A$-streams: and in this way, we re-find the familiar construction of
the final $\mathbb{T}_A$-comodel as the set of $A$-streams
under~\eqref{eq:3}.

The comodel view also allows us to capture the \emph{topology} on the
space of streams. Indeed, any comodel of a theory has a natural
prodiscrete topology (i.e., topologised as a product of discrete
spaces), whose basic open sets describe those states which are
indistinguishable with respect to a finite set of
$\mathbb{T}$-computations. (This definition appears to be novel.)

\begin{defi}[Operational topology]
  \label{def:6}
  Let $\alg[S]$ be a $\mathbb{T}$-comodel. The \emph{operational
    topology} on $S$ is generated by sub-basic open sets
  \begin{equation*}
    [t \mapsto v] := \{s \in S : \dbr{t}^{\alg[S]}(s) = (v, s') \text{ for some } s'\} \qquad \qquad \text{for all $t \in T(V)$ and $v \in V$}\rlap{ .}
  \end{equation*}
\end{defi}

\begin{lem}
  \label{lem:1}
  The operational topology makes any $\mathbb T$-comodel into a
  topological comodel. In the case of the final $\mathbb T$-comodel,
  this yields the \emph{final} topological comodel.
\end{lem}
Of course, a topological comodel is simply a comodel in the category
$\cat{Top}$ of topological spaces and continuous maps.
\begin{proof}
  Let $\alg[S]$ be a $\mathbb{T}$-comodel. Each co-operation
  $\dbr{\sigma}^{\alg[S]} \colon S \rightarrow \abs{\sigma} \cdot S$
  is continuous for the operational topology: for indeed, a sub-basic
  open set of the codomain is $\{i\} \times [t \mapsto v]$, and its
  inverse image under $\dbr{\sigma}^{\alg[S]}$ is the open set
  \begin{equation}\label{eq:42}
    \{\, s \in S \mid \exists s', s''.\, \dbr{\sigma}^{\alg[S]}(s) = (i, s') \text{ and }
    \dbr{t}^{\alg[S]}(s') = (v, s'')\} = [\sigma \mapsto i] \cap [(\sigma \gg t) \mapsto v]\text{ ,}
  \end{equation}
  where the equality comes from the fact that, if $\dbr{\sigma}^{\alg[S]}(s) = (i,s')$,
  then $\dbr{t}^{\alg[S]}(s') = \dbr{\sigma \gg t}^{\alg[S]}(s)$. 
  
  So the operational topology makes each comodel $\alg[S]$ into a
  topological comodel. We now show that, in the case of the final
  comodel $\alg[F]$, this yields the final topological comodel.
  Indeed, if $\alg[S]$ is \emph{any} topological comodel, we have by
  finality of $\alg[F]$ \emph{qua} $\cat{Set}$-comodel a unique
  comodel homomorphism $\beta_{(\thg)} \colon S \rightarrow F$ sending
  $s$ to the admisssible behaviour $\beta_s$, and we need only show
  this is continuous. But the inverse image under $\beta_{(\thg)}$ of
  the subbasic open $[t \mapsto v] \subseteq F$ is the set
  $\{s \in S : \beta_s(t) = v\} = \{s \in S : \dbr{t}^{\alg[S]}(s) =
  (v,s') \text{ for some $s' \in S$}\}$, which is open as the inverse
  image of $\{v\} \times S \subseteq V \times S$ under the continuous
  map $\dbr{t}^{\alg[S]} \colon S \rightarrow V \times S$.
\end{proof}

In the case of the theory of $A$-valued input, the subbasic open set
$[t \mapsto v] \subseteq A^\mathbb{N}$ of the final comodel can be
defined by induction on $t \in TV$:
\begin{equation*}
  [w \mapsto v] =
  \begin{cases}
    A^\mathbb{N} & \text{ if $w = v$}\\
    \emptyset & \text{ if $w \neq v \in V$}
  \end{cases} \quad \text{and} \quad [\mathsf{read}(\lambda a.\, t_a) \mapsto v] = \{ aW : a \in A, W \in [t_a \mapsto v]\}\rlap{ .}
\end{equation*}
From this description, we re-find the fact that the final topological
comodel is $A^\mathbb{N}$ endowed with the product topology for
$\mathbb{N}$ copies of the discrete space $A$.

\section{Stream processors as residual comodels}
\label{sec:runn-resid-runn}

In this section, we recall a more general kind of comodel considered
by, among others,~\cite{Ahman2020Runners, Katsumata2020Interaction,
  Uustalu2020Algebraic}, which allows for stateful translations
between different notions of computation. We then explain how this
notion allows us to encode stream processors in the sense
of~\cite{Hancock2009Representations} and also their
\emph{composition}.

\begin{defi}[Residual comodel] Let $\mathbb{T}$ and $\mathbb{R}$
  be theories. An \emph{$\mathbb{R}$-residual $\mathbb{T}$-comodel} is
  a comodel of $\mathbb{T}$ in the Kleisli category
  $\cat{Kl}(\mathbb{R})$.
\end{defi}

The nomenclature ``residual'' comes
from~\cite[\S5.3]{Katsumata2020Interaction}, and we will explain the
connection to \emph{loc.\,cit.} in Proposition~\ref{prop:3} below. For
now, let us spell out in detail what a residual comodel $\alg[S]$
involves. First, there is an underlying set of states $S$. Next, we
have for each operation $\sigma$ in the signature of $\mathbb{T}$ a basic co-operation
$\dbr{\sigma}^{\alg[S]} \colon S \rightarrow R(\abs \sigma \times S)$
assigning to each state $s \in S$ an $\mathbb{R}$-computation
$\dbr{\sigma}^{\alg[S]}\!(s)$ returning values in
$\abs \sigma \times S$---where, as before, we think of these two
components as providing a value answering the request posed by
$\sigma$, and a next state. Now we recursively determine a derived
co-interpretation
$\dbr{t}^{\alg[S]} \colon S \rightarrow R(V \times S)$ for each
$t \in T(V)$ via:
\begin{equation}\label{eq:4}
\begin{aligned}
  v \in V \subseteq T(V) &\implies \dbr{v}^{\alg[S]} (s) = (v,s) \in V \times S \subseteq R(V \times S) \qquad \ \ \,
  \\ \text{and }\sigma 
  \in \Sigma, t \in T(V)^{\abs \sigma} &\implies
  \dbr{\sigma(t)}^{\alg[S]}\!(s)= \dbr{\sigma}^{\alg[S]}\!(s)\bigl(\lambda (i, s').\, \dbr{t_i}^{\alg[S]}\!(s')\bigr)\rlap{ ,}
\end{aligned}
\end{equation}
and the final requirement is that these derived operations must
satisfy the equations of $\mathbb{T}$.


\begin{rem}
  \label{rk:5}
  In the second line of~\eqref{eq:4}, the element
  $\dbr{\sigma}^{\alg[S]}(s) \in R(\abs \sigma \times S)$ is an
  $\mathbb{R}$-term with variables in $\dbr{\sigma} \times S$; and
  substituting each variable $(i,s') \in \dbr{\sigma} \times S$
  therein by the $\mathbb{R}$-term $\dbr{t_i}^{\alg[S]}(s') \in R(V)$ gives
  the value of $\dbr{\sigma(t)}^{\alg[S]}(s) \in R(V)$. This
  amounts to threading
  $\mathbb{R}$-computations together by \emph{monadic binding}; in
  Haskell notation, we would write:
  \begin{equation*}
    \dbr{\sigma(t)}^{\alg[S]}(s) =
    \begin{minipage}[t]{4cm}
      \begin{flushleft}
      \textsf{\textbf{do}} $(i, s') \leftarrow \dbr{\sigma}^{\alg[S]}(s)$\\
      \phantom{\textsf{do}} \textsf{return} $\dbr{t_i}^{\alg[S]}(s')$.
      \end{flushleft}
    \end{minipage}
  \end{equation*}

\end{rem}

\begin{exa}
  A comodel of the theory $\mathbb{T}_B$ of $B$-valued input residual
  on the theory $\mathbb{P}_f^+$ of non-deterministic choice comprises
  a set of states $S$, and a function
  $\gamma \colon S \rightarrow P_f^+(B \times S)$: thus, a
  non-terminating, finitely branching labelled transition system.
\end{exa}

\begin{exa}
  A comodel of the theory $\mathbb{T}_B$ of $B$-valued input residual
  on the theory $\mathbb{D}$ of probabilistic choice comprises a set
  of states $S$, and a function
  $\gamma \colon S \rightarrow D(B \times S)$: thus, a finitely
  branching probabilistic generative system in the sense
  of~\cite{Glabbeek1990Reactive}.
\end{exa}

However, for us the key example is the following one:
\begin{exa}
  A comodel of the theory $\mathbb{T}_B$ of $B$-valued input residual
  on the theory $\mathbb{T}_A$ of $A$-valued input comprises a set of
  states $S$, and a function
  $\gamma \colon S \rightarrow T_A(B \times S)$ assigning to each
  state $s \in S$ a program which uses some number of $A$-tokens from
  an input stream to inform the choice of an output $B$-token and a
  new state in $S$.
\end{exa}

It is easy to see how each state $s_0$ of such a comodel should encode
a stream processor $A^\mathbb{N} \rightarrow B^\mathbb{N}$: given an
input stream $\vec a \in A^\mathbb{N}$, we consume some initial
segment $a_0, \dots, a_k$ to answer the requests posed by the program
$\gamma(s_0)$, so obtaining an element $b_0 \in B$ and a new state
$s_1$. We now repeat starting from $s_1 \in S$ and the remaining part
$\partial^k\vec a$ of the input stream, to obtain $b_1$ and $s_2$
while consuming $a_{k+1}, \dots, a_\ell$; and so on coinductively.
This description was made mathematically precise
in~\cite[\S3.1]{Hancock2009Representations}, but in fact we can obtain
it in a principled comodel-theoretic manner via (a special case
of) a notion given in~\cite[Appendix]{Plotkin2008Tensors}.

\begin{defi}[Tensor of a residual comodel with a comodel]
  \label{def:4}
  Let $\mathbb{T}$ and $\mathbb{R}$ be theories. Let $\alg[S]$ be an
  $\mathbb{R}$-residual $\mathbb{T}$-comodel, and let $\alg[M]$ be an
  $\mathbb{R}$-comodel. The \emph{tensor product} $\alg[S] \cdot \alg[M]$ is the $\mathbb{T}$-comodel with underlying set $S \times M$ and
  co-operations
  \begin{equation}\label{eq:5}
    \dbr{\sigma}^{\alg[S] \cdot \alg[M]} \colon S \times M \xrightarrow{\dbr{\sigma}^{\alg[S]} \times M}
    R(\abs{\sigma} \times S) \times M \xrightarrow{(t,m) \mapsto \dbr{t}^{\alg[M]}(m)} \abs \sigma \times S \times M\rlap{ .}
  \end{equation}
\end{defi}

This definition makes intuitive sense: given a state machine for
translating $\mathbb{T}$-computations into $\mathbb{R}$-computations,
and one for executing $\mathbb{R}$-computations, it threads them
together to yield a state machine for executing
$\mathbb{T}$-computations. We will make this justification rigorous in
Definition~\ref{def:2} below, but for the moment let us simply assume
its reasonability and give:

\begin{defi}[Trace]
  \label{def:5}
  Let $\alg[S]$ be a $\mathbb{T}_A$-residual $\mathbb{T}_B$-comodel.
  The \emph{trace} of a state $s \in \alg[S]$ is the function
  ${\mathsf{tr}(s) \colon A^\mathbb{N} \rightarrow B^\mathbb{N}}$
  obtained by partially evaluating at $s$ the unique map of
  $\mathbb{T}_B$-comodels
  $\alg[S] \cdot \alg[A^\mathbb{N}] \rightarrow
  \alg[B^\mathbb{N}]$, where $\alg[A^\mathbb{N}]$ and
  $\alg[B^\mathbb{N}]$ are endowed with their final comodel
  structures.
\end{defi}

We now unfold this definition. Firstly, for any term $t \in T_A(V)$,
the derived co-operation
$\dbr{t}^{\alg[A^\mathbb{N}]}\!(\vec a) \colon A^\mathbb{N}
\rightarrow V \times A^\mathbb{N}$ is defined recursively by
\begin{equation}\label{eq:18}
  \dbr{v}^{\alg[A^\mathbb{N}]}\!(\vec a) = (v, \vec a) \qquad \qquad \text{and} \qquad \qquad \dbr{\mathsf{read}(t)}^{\alg[A^\mathbb{N}]}\!(\vec a) = \dbr{t_{a_0}}^{\alg[A^\mathbb{N}]}\!(\partial \vec a)\rlap{ .}
\end{equation}
If we view $t$ as an $A$-ary branching tree with leaves labelled in
$V$, then $\dbr{t}^{\alg[A^\mathbb{N}]}\!(\vec a)$ is the result of
walking up the tree from the root, consuming an element of $\vec a$ at
each interior node to determine which branch to take, and returning at
a leaf the $V$-value found there along with what remains of $\vec a$.

Now, in terms of this, the $\mathbb{T}_B$-comodel structure of
$\alg[S] \cdot \alg[A^\mathbb{N}]$ is given by
\begin{equation*}
  S \times A^\mathbb{N} \rightarrow B \times S \times A^\mathbb{N} \qquad\qquad\qquad
  (s, \vec a)  \mapsto \dbr{\gamma(s)}^{\alg[A^\mathbb{N}]}\!(\vec a)\rlap{ ,}
\end{equation*}
where $\gamma \colon S \rightarrow T_A(B \times S)$ is the residual
comodel structure of $\alg[S]$. 
This function takes a state $s_0$ and stream $\vec a$ to
the triple $(b_0, s_1, \partial^k \vec a)$ obtained by walking up $k$
nodes of the tree $\gamma(s)$ to the leaf $(b_0, s_1)$. If we view
this comodel structure as a triple of maps
\begin{equation*}
\mathsf{hd} \colon S \times A^\mathbb{N} \rightarrow B \qquad \mathsf{next} \colon S \times A^\mathbb{N} \rightarrow S \qquad \mathsf{tl} \colon S \times A^\mathbb{N} \rightarrow A^\mathbb{N}
\end{equation*}
then we can say, finally, that the trace
$\mathsf{tr}(s) \colon A^\mathbb{N} \rightarrow B^\mathbb{N}$ of $s
\in \alg[S]$
is given coinductively by:
\begin{equation*}
  \bigl(\mathsf{tr}(s)(\vec a)\bigr)_0 = \mathsf{hd}(s, \vec a) \qquad \qquad \partial \bigl(\mathsf{tr}(s)(\vec a)\bigr) = \mathsf{tr}(\mathsf{next}(s, \vec a))(\mathsf{tl}(s, \vec a))\rlap{ .}
\end{equation*}
Comparing this construction with that
of~\cite[\S3.1]{Hancock2009Representations}, done there with bare
hands, we find that they are exactly the same: the derived
co-operations $\dbr{t}$ of~\eqref{eq:18} are the
functions $eat\ t$ of~\emph{loc.\,cit.}, while our trace function
$\mathsf{tr}$ is their function $eat_\infty$. 

We have thus shown that each state $s$ of a $\mathbb{T}_A$-residual
$\mathbb{T}_B$-comodel encodes a function
$\mathsf{tr}(s) \colon A^\mathbb{N} \rightarrow B^\mathbb{N}$; but for
these functions to be reasonable stream processors, they should be
\emph{continuous} for the profinite topologies. While this may be
shown with little effort, we may in fact see it without \emph{any} effort via
a comodel-theoretic argument. We first need:

\begin{defi}[Tensor of a residual comodel and a topological
  comodel]
  \label{def:toptensor}
  Let $\mathbb{T}$ and $\mathbb{R}$ be theories, let $\alg[S]$ be an
  $\mathbb{R}$-residual $\mathbb{T}$-comodel, and $\alg[M]$ a
  topological $\mathbb{R}$-comodel. The \emph{tensor product}
  $\alg[S] \cdot \alg[M]$ is the topological $\mathbb{T}$-comodel
  with underlying space $S \cdot M$ and
  co-operations~\eqref{eq:5}.
\end{defi}

Once again, the justification for this definition will be given below;
assuming it for now, the desired continuity of each $\mathsf{tr}(s)$
is immediate. For indeed, viewing $\alg[A^\mathbb{N}]$ and
$\alg[B^\mathbb{N}]$ as final topological comodels with the profinite
topology, there is a unique map of topological $\mathbb{T}_B$-comodels
$\alg[S] \cdot \alg[A^\mathbb{N}] \rightarrow \alg[B^\mathbb{N}]$. Its
underlying function is the unique map of $\cat{Set}$-comodels from
Definition~\ref{def:5}, but the extra information we now gain is the
\emph{continuity} of this map---which says precisely that each
$\mathsf{tr}(s) \colon A^\mathbb{N} \rightarrow B^\mathbb{N}$ is
continuous, as desired.

The comodel perspective allows us also to say something about
\emph{composition} of stream processors. Given a $\mathbb{T}_A$-residual
$\mathbb{T}_B$-comodel $\alg[S]$, whose states encode the continuous
functions $\mathsf{tr}(s) \colon A^\mathbb{N} \rightarrow
B^\mathbb{N}$ for each $s \in \alg[S]$, and a $\mathbb{T}_B$-residual
$\mathbb{T}_C$-comodel $\alg[P]$, whose states encode the continuous
functions $\mathsf{tr}(p) \colon B^\mathbb{N} \rightarrow
C^\mathbb{N}$ for each $p \in \alg[P]$, we may define the \emph{tensor
  product} $\alg[P] \cdot \alg[S]$, which is a
$\mathbb{T}_A$-residual $\mathbb{T}_C$-comodel whose states encode
precisely the continuous functions
$\mathsf{tr}(p) \circ \mathsf{tr}(s) \colon A^\mathbb{N} \rightarrow
C^\mathbb{N}$ for $s \in \alg[S]$ and $p \in \alg[P]$. The general
definition is as follows; again, this will be justified formally by
Definition~\ref{def:2} below.



\begin{defi}[Tensor product of two residual comodels]
  \label{def:twocomodels}
  Let $\mathbb{T}$, $\mathbb{R}$ and $\mathbb{V}$ be theories,
  let $\alg[S]$ be an $\mathbb{R}$-residual $\mathbb{T}$-comodel
  and let $\alg[P]$ be a $\mathbb{T}$-residual
  $\mathbb{V}$-comodel. The tensor product $\alg[P]\cdot \alg[S]$
is the $\mathbb{T}$-residual $\mathbb{V}$-comodel with underlying set
$P \times S$ and co-operations
\begin{equation*}
  \dbr{\sigma}^{\alg[P] \cdot \alg[S]} \colon P \times S \xrightarrow{\dbr{\sigma}^{\alg[P]} \times S}
  T(\abs{\sigma} \times P) \times S \xrightarrow{(t,s) \mapsto \dbr{t}^{\alg[S]}(s)} T(\abs \sigma \times P \times S)\rlap{ .}
\end{equation*}
\end{defi}

When in this definition, $\alg[S]$ is a $\mathbb{T}_A$-residual
$\mathbb{T}_B$-comodel and $\alg[P]$ is a $\mathbb{T}_B$-residual
$\mathbb{T}_C$-comodel, the tensor $\alg[P] \cdot \alg[S]$ is the
$\mathbb{T}_A$-residual $\mathbb{T}_C$-comodel with underlying set $P
\times S$ and structure map
\begin{equation}\label{eq:22}
  \dbr{\mathsf{read}}^{\alg[P] \cdot \alg[S]} \colon P \times S \xrightarrow{\dbr{\mathsf{read}}^{\alg[P]} \times S}
  T_B(C \times P) \times S \xrightarrow{(t,s) \mapsto \dbr{t}^{\alg[S]}(s)} T_{A}(C \times P \times S)\rlap{ .}
\end{equation}
To understand this, we must now unfold the definition of
$\dbr{t}^{\alg[S]}$, which is given by structural recursion over
$T_B(C \times P)$ as in~\eqref{eq:4}:
\begin{equation}
  \label{eq:23}
\begin{aligned}
  \dbr{(c, p)}^{\alg[S]}(s) & = (c, p, s) & & \text{for $(c, p) \in C \times P \subseteq T_B(C \times P)$}\\
  \dbr{\mathsf{read}(\lambda b.\, t_b)}^{\alg[S]}(s) &=   \dbr{\mathsf{read}}^{\alg[S]}(s)\bigl(\lambda (b, s').\, \dbr{t_b}^{\alg[S]}(s')\bigr) & & \text{for $t \in T_B(C \times P)^B$.}
\end{aligned}
\end{equation}
Here, in the second clause, $\dbr{\mathsf{read}}^{\alg[S]}(s)$ is a
term in $T_{A}(B \times S)$, into which we are substituting the
$B \times S$-indexed family of terms
$\dbr{t_b}^{\alg[S]}(s') \in T_{A}(C \times P \times S)$ to obtain the
desired term in $T_{A}(C \times P \times S)$.

Let us now see that the states of $\alg[P] \cdot \alg[S]$ encode
precisely the composites of the continuous functions encoded by the
states of $\alg[P]$ and $\alg[S]$. Despite the complexity of our
description of $\alg[P] \cdot \alg[S]$, the proof of this fact is
trivial.

\begin{prop}\label{prop:spr}
  Suppose that $\alg[S]$ is a $\mathbb{T}_A$-residual
  $\mathbb{T}_B$-comodel and $\alg[P]$ is a $\mathbb{T}_B$-residual
  $\mathbb{T}_C$-comodel. For any $s \in \alg[S]$ and $p \in \alg[P]$,
  we have that $\mathsf{tr}^{\alg[P] \cdot \alg[S]}(p,s) =
  \mathsf{tr}^{\alg[P]}(p) \circ \mathsf{tr}^{\alg[S]}(s) \colon
  A^\mathbb{N} \rightarrow C^\mathbb{N}$.
\end{prop}
\begin{proof}
  Consider the diagram of $\mathbb{T}_C$-comodels
  \begin{equation*}
    \cd[@-1em@C-1em]{
      \alg[P] \cdot \alg[S] \cdot \alg[A^\mathbb{N}] \ar[rr]^-{\alg[P] \cdot !} \ar[dr]_-{!} & &
      \alg[P] \cdot \alg[B^\mathbb{N}] \ar[dl]^-{!} \\ & \alg[C^\mathbb{N}]
    }
  \end{equation*}
  wherein each map labelled $!$ is the unique map into a final object.
  Since $\alg[C^\mathbb{N}]$ is final among $\mathbb{T}_C$-comodels,
  this diagram clearly commutes; and now partially evaluating at
  $(p,s) \in P \times S$ yields the desired equality.
\end{proof}

Before continuing, we resolve some unfinished business by justifying
Definitions~\ref{def:4}, \ref{def:toptensor} and~\ref{def:twocomodels}
above. Our starting point will be an alternative presentation of the
notion of comodel due to~\cite{Uustalu2015Stateful}. In
\emph{op.\,cit.}, Uustalu defines a \emph{runner} for a theory
$\mathbb{T}$, with set of states $S$, to be a monad morphism
$\mathsf{T} \rightarrow \mathsf{T}_S$ from the associated monad of
$\mathbb{T}$ to the state monad $\mathsf{T}_S = (\thg \times S)^S$.
The data for such a runner are functions
$T(V) \rightarrow (V \times S)^S$ assigning to each $t \in T(V)$ a
function $\dbr{t}^{\alg[S]} \colon S \rightarrow V \times S$.
Recognising these as the data of the derived co-operations
of a $\mathbb{T}$-comodel structure on $S$, we should find the main
result of~\cite{Uustalu2015Stateful} reasonable: that
$\mathbb{T}$-comodels with underlying set $S$ are in bijection with
$\mathbb{T}$-runners with underlying set of states $S$.

While Uustalu's result is about comodels in $\cat{Set}$, it
generalises unproblematically. For any object $S$ of a category $\C$
with copowers, we have an adjunction
$(\thg) \cdot S \dashv \C(S, \thg) \colon \C \rightarrow \cat{Set}$
inducing a monad $\mathsf{T}_S = \C(S, (\thg) \cdot S)$ on
$\cat{Set}$; in~\cite{Mogelberg2014Linear} this is called the
\emph{linear-use state monad} associated to $S$. We now have the
following natural extension of Uustalu's result.

\begin{propC}[{\cite[Theorem~8.2]{Mogelberg2014Linear}}]
  \label{prop:2}
  Let $\mathbb{T}$ be an algebraic theory, $\C$ a category with
  copowers, and $S \in \C$. The following are in bijective
  correspondence:
  \begin{enumerate}
  \item $\mathbb{T}$-comodels $\alg[S]$ in $\C$ with underlying object $S$;
  \item $\mathbb{T}$-runners in $\C$, i.e., monad maps
    $\dbr{\thg}^{\alg[S]} \colon \mathsf{T} \rightarrow \mathsf{T}_S$ into the
    linear-use state monad of $S$;
  \item Functorial extensions of 
    $(\thg) \cdot S \colon \cat{Set} \rightarrow \C$ along the
    free functor into the Kleisli category:
  \begin{equation}\label{eq:2}
    \cd[@-0.5em]{
      \cat{Set} \ar[r]^-{(\thg) \cdot S} \ar[d]_-{F_\mathbb{T}} & \C \rlap{ .}\\
      \cat{Kl}(\mathbb{T}) \ar@{-->}[ur]_-{(\thg) \cdot \alg[S]}
    }
  \end{equation}
  \end{enumerate}
\end{propC}
\begin{rem}
  \label{rk:3}
  Abstractly, this proposition expresses the fact that $\cat{Kl}(\mathbb{T})$ is the
  \emph{free category with copowers containing a comodel of
    $\mathbb{T}$}; this result is originally due to
  Linton~\cite{Linton1966Some}.
\end{rem}
\begin{proof}
  As just said, the argument for (1) $\Leftrightarrow$ (2) is
  \emph{mutatis mutandis} that of~\cite[\S 3]{Uustalu2015Stateful}.
  For (2) $\Leftrightarrow$ (3), it is
  standard~\cite{Meyer1975Induced} that monad maps
  $\mathsf{T} \rightarrow \mathsf{T}_S$ correspond to extensions to
  the left~in:
  \begin{equation*}
    \cd[@-0.5em]{
      \cat{Set} \ar[r]^-{F^{\mathsf{T}_S}} \ar[d]_-{F^\mathsf{T}} & \cat{Kl}(\mathsf{T}_S) \\
      \cat{Kl}(\mathsf{T}) \ar@{-->}[ur]_-{}
    } \qquad \qquad
        \cd[@-0.5em]{
      \cat{Set} \ar[r]^-{(\thg) \cdot S} \ar[d]_-{F_\mathbb{T}} & \C_S\rlap{ .} \\
      \cat{Kl}(\mathbb{T}) \ar@{-->}[ur]_-{}
    }
  \end{equation*}
  Now the Kleisli category $\cat{Kl}(\mathsf{T}_S)$ of the linear-use
  state monad is isomorphic to the category $\C_S$ whose
  objects are sets, and whose maps $A \rightarrow B$ are $\C$-maps
  $A \cdot S \rightarrow B \cdot S$, via an isomorphism which
  identifies $F^{\mathsf{T}_S}$ with
  $(\thg) \cdot S \colon \cat{Set} \rightarrow \C_S$. Similarly we have
  $\cat{Kl}(\mathsf{T}) \cong \cat{Kl}(\mathbb{T})$ under $\cat{Set}$. So monad
  morphisms $\mathsf{T} \rightarrow \mathsf{T}_S$ correspond to
  extensions as right above: and these, by direct inspection, correspond to
  extensions as in~\eqref{eq:2}.
\end{proof}

\looseness=-1 If here $\C$ is itself the Kleisli category
$\cat{Kl}(\mathbb{R})$ of a theory $\mathbb{R}$, then the linear-use
state monad of $S \in \cat{Kl}(\mathbb{R})$ is the monad
$\mathsf{R}(\thg \times S)^S$ found as the commuting combination of
the state monad for $S$ with the monad $\mathsf{R}$ induced by
$\mathbb{R}$ (cf.~\cite[Theorem~10]{Hyland2006Combining}). Monad
maps $\mathsf{T} \rightarrow \mathsf{R}(S \times \thg)^S$ were
in~\cite{Katsumata2020Interaction} termed \emph{$\mathbb{R}$-residual
  $\mathbb{T}$-runners}, and for these the preceding result
specialises to:
\begin{prop}
  \label{prop:3}
  Let $\mathbb{T}$ and $\mathbb{R}$ be algebraic theories and let $S$
  be a set. The following are in bijective correspondence:
  \begin{enumerate}
  \item $\mathbb{R}$-residual $\mathbb{T}$-comodels $\alg[S]$ with
    underlying set $S$;
  \item $\mathbb{R}$-residual $\mathbb{T}$-runners $\dbr{\thg}^{\alg[S]} \colon
    \mathsf{T} \rightarrow \mathsf{R}(\thg \times S)^S$;
  \item Functorial extensions of
    $(\thg) \times S \colon \cat{Set} \rightarrow \cat{Set}$ through
    the Kleisli categories of $\mathbb{T}$ and $\mathbb{R}$:
  \begin{equation}\label{eq:7}
    \cd[@-0.5em@C+0.2em]{
      \cat{Set} \ar[r]^-{(\thg) \times S} \ar[d]_-{F_\mathbb{T}} & \cat{Set} \ar[d]^-{F_\mathbb{R}} \\
      \cat{Kl}(\mathbb{T}) \ar@{-->}[r]^-{(\thg) \cdot \alg[S]} &
      \cat{Kl}(\mathbb{R})\rlap{ .}
    }
  \end{equation}
  \end{enumerate}
\end{prop}

By putting together Propositions~\ref{prop:2} and~\ref{prop:3}, we
have an intuitive definition of \emph{tensor product} for a residual
comodel and a comodel, or for two residual comodels.
\begin{defi}[Tensor product of residual comodels]
  \label{def:2}\looseness=-1
  Let $\mathbb{V}$, $\mathbb{T}$, $\mathbb{R}$  be theories;
  $\alg[M]$ an $\mathbb{R}$-comodel in $\C$; $\alg[S]$ an
  $\mathbb{R}$-residual $\mathbb{T}$-comodel; and $\alg[P]$ a
  $\mathbb{T}$-residual $\mathbb{V}$-comodel. The tensor
    product $\alg[S] \cdot \alg[M]$ is the $\mathbb{T}$-comodel in
  $\C$ classified by the composite of extensions to the left
  below, while the tensor product $\alg[P] \cdot \alg[S]$ is
  the $\mathbb{R}$-residual $\mathbb{V}$-comodel classified by the
  composite to the right:
  \begin{equation}\label{eq:1}
    \cd[@C+0.1em@-0.3em]{
      \cat{Set} \ar[r]^-{(\thg) \times S} \ar[d]_-{F_\mathbb{T}} &
      \cat{Set} \ar[r]^-{(\thg) \cdot M} \ar[d]_-{F_\mathbb{R}} & \C & &
      \cat{Set} \ar[r]^-{(\thg) \times P} \ar[d]_-{F_\mathbb{V}} &
      \cat{Set} \ar[r]^-{(\thg) \times S} \ar[d]_-{F_\mathbb{T}} &
      \cat{Set} \ar[d]^-{F_\mathbb{R}} & &
      \\
      \cat{Kl}(\mathbb{T}) \ar@{-->}[r]^-{(\thg) \cdot \alg[S]} &
      \cat{Kl}(\mathbb{R}) \ar@{-->}[ur]_-{(\thg) \cdot \alg[M]} & & & 
      \cat{Kl}(\mathbb{V}) \ar@{-->}[r]^-{(\thg) \cdot \alg[P]} &
      \cat{Kl}(\mathbb{T}) \ar@{-->}[r]^-{(\thg) \cdot \alg[S]} &
      \cat{Kl}(\mathbb{R}) \rlap{ .}
    }
  \end{equation}
\end{defi}
In particular, when $\C = \cat{Set}$ and $\C = \cat{Top}$, the tensor
product $\alg[S] \cdot \alg[M]$ specialises to those of
Definitions~\ref{def:4} and~\ref{def:toptensor} above; while the
tensor product $\alg[P] \cdot \alg[S]$ yields Definition~\ref{def:twocomodels}.

\begin{rem}
  \label{rk:4}
  Here is another perspective on Definition~\ref{def:2}. To the left
  of~(\ref{eq:1}), the functor $(\thg) \cdot \alg[M]$ preserves
  copowers, and so lifts to a functor
  $\cat{Comod}(\mathbb{T}, \cat{Kl}(\mathbb{R})) \rightarrow
  \cat{Comod}(\mathbb{T}, \C)$, whose value at $\alg[S]$ is the tensor
  product $\alg[S] \cdot \alg[M]$. We can obtain
  $\alg[P] \cdot \alg[S]$ to the right similarly.
\end{rem}

\section{Intensional stream processors as a final residual comodel}
\label{sec:intens-stre-proc}
The arguments of the previous section were given for an arbitrary
$\mathbb{T}_A$-residual $\mathbb{T}_B$-comodel $\alg[S]$; but as
in~\cite{Hancock2009Representations}, it is natural to consider the
final residual comodel in particular. To this end, we should first
clarify the correct notion of \emph{morphism} between residual
comodels.

\begin{defi}[Map of residual comodels]
  \label{defn:17}
  Let $\mathbb{T}$ and $\mathbb{R}$ be theories, and let $\alg[S]$ and
  $\alg[U]$ be $\mathbb{R}$-residual $\mathbb{T}$-comodels. A map of
  residual comodels $\alg[S] \to \alg[U]$ is a function
  $f \colon S \rightarrow U$ such that
  $\dbr{\sigma}^{\alg[U]} \circ f = R(\abs \sigma \times f) \circ
  \dbr{\sigma}^{\alg[S]}$ for all operations $\sigma$ in the signature
  of $\mathbb{T}$.
\end{defi}
\begin{rem}
  \label{rk:1}
  Given that an $\mathbb{R}$-residual comodel is a comodel in
  $Kl(\mathbb{R})$, we might expect a map of residual comodels to be a
  map in $Kl(\mathbb{R})$, rather than one in $\cat{Set}$. The reason
  for our choice is not pure expediency; it has to do with an
  enrichment of the category of theories in the category of comonads
  on $\cat{Set}$, currently being investigated by the authors
  of~\cite{Katsumata2020Interaction}, and which exploits the general
  \emph{Sweedler theory} of~\cite{Anel2013Sweedler}. Working through
  the calculations, one finds that for two theories $\mathbb{R}$ and
  $\mathbb{T}$, the category of coalgebras for the hom-comonad
  $\langle\mathbb{R}, \mathbb{T}\rangle$ is the category of residual
  comodels, with precisely the maps indicated in
  Definition~\ref{defn:17}.
\end{rem}

With this clarification made, we see that, in particular, the category
of $\mathbb{T}_A$-residual $\mathbb{T}_B$-comodels is simply the
category of $T_A(B \times \thg)$-coalgebras, and so we have:
\begin{defi}[Intensional stream processors]
  \label{def:9}
  The \emph{type of intensional $A$-$B$-stream processors} is the
  final $\mathbb{T}_A$-residual $\mathbb{T}_B$-comodel
  $\alg[{I_{AB}}]$, i.e., the final $T_A(B \times \thg)$-coalgebra
  \begin{equation}\label{eq:17}
    \theta_{AB} \colon I_{AB} \rightarrow T_A(B \times I_{AB})\rlap{ .}
  \end{equation}
  The \emph{reflection} function is the
  trace function of the final residual comodel:
  \begin{equation*}
    \mathsf{reflect} \colon I_{AB} \rightarrow \cat{Top}(A^\mathbb{N}, B^\mathbb{N}) \qquad \qquad s \mapsto \mathsf{tr}(s) \colon A^\mathbb{N} \rightarrow B^\mathbb{N}\rlap{ ,}
  \end{equation*}
  where here we write $\cat{Top}$ for the category of topological
  spaces and continuous maps.
\end{defi}

As well as reflection, \cite{Hancock2009Representations} also defines
a \emph{reification} function
$\mathsf{reify} \colon \cat{Top}(A^\mathbb{N}, B^\mathbb{N})
\rightarrow I_{AB}$ that implements each continuous function by a
state of the final comodel, and which satisfies
$\mathsf{reflect} \circ \mathsf{reify} = \mathrm{id}$. This means that
$\mathsf{reflect}$ is surjective---but crucially, it is \emph{not}
injective. To show this, we must first note that, by the usual
techniques, the terminal coalgebra $I_{AB}$ may be described as
follows: it is the set of all finite or infinite $A$-ary branching
trees, with interior nodes labelled with elements of $B^\ast$ (i.e.,
lists of elements of $B$), with leaves labelled by elements of
$B^\mathbb{N}$, and where no infinite path of interior nodes is
labelled by the empty list.
\begin{exa}
  \label{ex:3}  
  Fix an element $b \in B$ and consider the following two
  $\mathbb{T}_A$-residual $\mathbb{T}_B$-comodel structures on
  $\{\ast\}$:
  \begin{equation}\label{eq:8}
    \qquad \text{(i) }\ast \mapsto (b, \ast) \qquad \qquad \text{and} \qquad \qquad
    \text{(ii) }\ast \mapsto \textsf{read}(\lambda a.\,(b, \ast))\rlap{ .}
  \end{equation}
  In both comodels, the unique state $\ast$ encodes the continuous
  function $A^\mathbb{N} \rightarrow B^\mathbb{N}$ sending every
  stream $\vec a$ to $(b,b,b,b,\dots)$. However, these states yield
  different elements of the final comodel $I_{AB}$: (i) gives the
  trivial tree $\tau_0$ whose root is labelled by $(b,b,b,\dots)$, while
  (ii) gives the purely infinite $A$-ary branching tree $\tau_1$ with
  every node labelled by a single $b$.
\end{exa}

Intuitively, the two states of $I_{AB}$ in this example differ in that
the first ignores its input stream entirely, and simply outputs $b$'s
without cease; while the second frivolously consumes a single
$A$-token before emitting each $b$. So $I_{AB}$ is a set of
\emph{intensional} representations of stream processors. This will
lead us neatly on to the second part of the paper, where we give a
comodel-theoretic presentation of \emph{extensional} stream
processors, i.e., the set $\cat{Top}(A^\mathbb{N}, B^\mathbb{N})$, and
an explanation in these terms of where the reification function
of~\cite{Hancock2009Representations} comes from.

Before we do this, let us see how the tensor product of residual
comodels allows us to give an account of the \emph{lazy composition}
of intensional stream processors
from~\cite[\S4]{Hancock2009Representations}.

\begin{defi}\label{def:lazy}
  The \emph{lazy composition} of intensional stream processors is
  given by the unique map of $\mathbb{T}_A$-residual
  $\mathbb{T}_C$-comodels
  $\mathsf{comp} \colon \alg[I_{BC}] \cdot \alg[I_{AB}] \rightarrow
  \alg[I_{AC}]$, where in the domain we take the tensor product of
  residual comodels of Definition~\ref{def:2}.
\end{defi}

Let us now unpack this definition to see that our composition agrees
with the one given in~\cite[\S4]{Hancock2009Representations}. First,
as a special case of~\eqref{eq:22}, the $\mathbb{T}_A$-residual
$\mathbb{T}_C$-comodel structure on $\alg[I_{BC}] \cdot \alg[I_{AB}]$
has co-operation
$\dbr{\mathsf{read}}^{\alg[I_{BC}] \cdot \alg[I_{AB}]}$ given by:
\begin{equation*}
  I_{BC} \times I_{AB} \xrightarrow{\theta_{BC} \times  \mathrm{id}}
  T_B(C \times I_{BC}) \times I_{AB} \xrightarrow{(t,\sigma) \mapsto \dbr{t}^{\alg[I_{AB}]}(\sigma)} T_{A}(C \times I_{BC} \times I_{AB})\rlap{ ,}
\end{equation*}
where, as in~\eqref{eq:17}, we write $\theta_{BC}$ for
$\dbr{\mathsf{read}}^{\alg[I_{BC}]}$. Furthermore, as
in~\eqref{eq:23}, we have $\dbr{t}^{\alg[I_{AB}]}$ defined recursively
by:
\begin{equation}
  \label{eq:25}
\begin{aligned}
  \dbr{(c, \tau)}^{\alg[I_{AB}]}(\sigma) & = (c, \tau, \sigma) && \text{for $(c, \tau) \in C \times I_{BC}$, $\sigma \in I_{AB}$}\\
  \dbr{\mathsf{read}(t)}^{\alg[I_{AB}]}(\sigma) &=   \theta_{AB}(\sigma)\bigl(\lambda (b, \sigma').\, \dbr{t_b}^{\alg[I_{AB}]}(\sigma')\bigr) && \text{for $t \in T_B(C \times I_{BC})^B$, $\sigma \in I_{AB}$.}
\end{aligned}
\end{equation}

We now compare this with the composition function of~\cite[\S
4.1]{Hancock2009Representations}, which was obtained as follows.
First, the authors define
a $\mathbb{T}_A$-residual $\mathbb{T}_C$-comodel structure
$\chi \colon S \rightarrow T_{A}(C \times S)$
on the set $S = T_{B}(C \times I_{BC}) \times T_{A}(B \times I_{AB})$,
via the following clauses:
\begin{align*}
  \chi\bigl((c,\tau), u) & = (c, \theta_{BC}(\tau), u) && \text{for $(c,\tau) \in C \times I_{BC}$}\\
  \chi\bigl(\mathsf{read}(t), (b, \sigma)\bigr) & = \chi(t_b, \theta_{AB}(\sigma)) && \text{for $t \in T_B(C \times I_{BC})^B$, $(b,\sigma) \in B \times I_{AB}$}\\
  \chi\bigl(\mathsf{read}(t), \mathsf{read}(u)\bigr) & = \mathsf{read}(\lambda a.\, \chi(\mathsf{read}(t), u_a)) && \text{for $t \in T_B(C \times I_{BC})^B$, $u \in T_A(B \times I_{AB})^A$;}
\end{align*}
they now induce by finality of $\alg[I_{AC}]$ a unique map of residual
comodels $u \colon \alg[S] \rightarrow \alg[I_{AC}]$; and finally, they define the
composition map $I_{BC} \times I_{AB} \rightarrow I_{AC}$ as $u
\circ (\theta_{BC} \times \theta_{AB})$.

\begin{prop}
  The lazy composition of Definition~\ref{def:lazy} coincides with
  that of~\cite{Hancock2009Representations}.
\end{prop}
\begin{proof}
  Given the definitions of the two maps, it suffices to show that
  $\theta_{BC} \times \theta_{AB}$ is a map of residual comodels
  $\alg[I_{BC}] \cdot \alg[I_{AB}] \rightarrow \alg[S]$, i.e., that
  the outside of
\begin{equation*}
  \cd[@C+3em]{
    I_{BC} \times I_{AB} \ar[r]^-{\theta_{BC} \times \mathrm{id}} \ar[d]_-{\theta_{BC} \times \theta_{AB}} &
    T_B(C \times I_{BC}) \times I_{AB} \ar[dl]^-{\mathrm{id} \times \theta_{AB}}\ar[r]^-{(t,\sigma) \mapsto \dbr{t}^{\alg[I_{AB}]}(\sigma)} &
    T_{A}(C \times I_{BC} \times I_{AB}) \ar[d]^-{T_A(C \times \theta_{BC} \times \theta_{AB})} \\
    S \ar[rr]_(0.55){\chi} & &
    T_A(C \times S)
  }
\end{equation*}
commutes. Clearly the left triangle commutes, so we need only check
the same for the right square. From~\eqref{eq:25}, the upper composite
$f \colon T_B(C \times I_{BC}) \times I_{AB} \rightarrow T_A(C \times
S)$ around this square is the map defined recursively by
\begin{equation*}
\begin{aligned}
  f\bigl((c, \tau),\sigma\bigr) & = (c, \theta_{BC}(\tau), \theta_{AB}(\sigma)) && \text{for $(c, \tau) \in C \times I_{BC}$, $\sigma \in I_{AB}$}\\
  f\bigl(\mathsf{read}(t),\sigma\bigr) &=   \theta_{AB}(\sigma)\bigl(\lambda (b, \sigma').\, f(t_b,\sigma')\bigr) && \text{for $t \in T_B(C \times I_{BC})^B$, $\sigma \in I_{AB}$.}
\end{aligned}
\end{equation*}
while by~\eqref{eq:20}, the clauses
defining $\chi$ can be rewritten as
\begin{equation*}
\begin{aligned}
  \chi\bigl((c,\tau), u) & = (c, \theta_{BC}(\tau), u) && \text{for $(c,\tau) \in C \times I_{BC}$, $u \in T_A(B \times I_{AB})$}\\
  \chi\bigl(\mathsf{read}(t), u\bigr) & = u(\lambda (b,\sigma).\, \chi(t_{b}, \theta_{AB}(\sigma))) && \text{for $t \in T_B(C \times I_{BC})^B$, $u \in T_A(B \times I_{AB})$.}
\end{aligned}
\end{equation*}
Comparing these two formulae, it now follows by structural induction
on $t$ that $f(t, \sigma) = \chi(t, \theta_{AB}(\sigma))$ for all
$(t, \sigma) \in T_{B}(C \times I_{BC}) \times I_{AB}$, as desired.
\end{proof}

We may now deduce, as in~\cite[\S 4.2]{Hancock2009Representations},
that composition of stream processsors corresponds to composition of
the underlying continuous functions. Much as in
Proposition~\ref{prop:spr}, the proof in our setting is trivial.

\begin{prop}
  The following diagram commutes for all $A,B,C$:
  \begin{equation}\label{eq:21}
    \cd{
      I_{BC} \times I_{AB} \ar[d]_-{\mathsf{reflect} \times \mathsf{reflect}} \ar[r]^-{\mathsf{comp}} & I_{AC} \ar[d]^-{\mathsf{reflect}} \\
      \mathsf{Top}(B^\mathbb{N}, C^\mathbb{N}) \times 
      \mathsf{Top}(A^\mathbb{N}, B^\mathbb{N}) \ar[r]^-{\circ} & \mathsf{Top}(A^\mathbb{N}, C^\mathbb{N})
    }
  \end{equation}
\end{prop}
\begin{proof}
  Consider the diagram of topological $\mathbb{T}_C$-comodels and
  comodel homomorphisms:
  \begin{equation*}
    \cd[@+1em]{
      \alg[I_{BC}] \cdot \alg[I_{AB}] \cdot \alg[A^\mathbb{N}] \ar[d]_-{\alg[I_{BC}] \cdot !} \ar[r]^-{! \cdot \alg[A^\mathbb{N}]} & \alg[I_{AC}] \cdot \alg[A^\mathbb{N}] \ar[d]^-{!} \\
      \alg[I_{BC}] \cdot \alg[B^\mathbb{N}] \ar[r]^-{!} & \alg[C^\mathbb{N}]
    }
  \end{equation*}
  where each ``$!$'' denotes a unique map to a final object. 
  Since $\alg[C^\mathbb{N}]$ is final, this diagram commutes, and
  currying around the two sides yields the corresponding two sides
  of~\eqref{eq:21}, which thus also commutes.
\end{proof}

Before continuing, let us note that other kinds of residual comodel
tensor product are also interesting in this context. For example, if
we are given a $\mathbb{P}_f^+$-residual $\mathbb{T}_A$-comodel
$\alg[S]$ (i.e., a labelled transition system with label-set $A$)
together with a $\mathbb{T}_A$-residual $\mathbb{T}_B$-comodel
$\alg[P]$, then we can tensor them together to get a
$\mathbb{P}_f^+$-residual $\mathbb{T}_B$-comodel
$\alg[P] \cdot \alg[S]$, i.e., a labelled transition system with
label-set $B$; and given a state $s \in \alg[S]$ and a state
$p \in \alg[P]$, we have the state $(p,s) \in \alg[P] \cdot \alg[S]$
in which the transition system $\alg[S]$ produces a stream of
$A$-labels starting from state $s$, and feeds them into the stream
processor $\alg[P]$ starting from state $p$ in order to produce a
stream of $B$-values. Just as before, we can calculate explicitly the
state machine $\alg[P] \cdot \alg[S]$, and see that it is \emph{lazy}
in the sense that $\alg[P]$ requests only the minimal possible number
of $A$-tokens from the transition system $\alg[S]$ in order to produce
each output $B$-token.

\section{Extensional stream processors as a final bimodel}
\label{sec:extens-stre-proc}

In Section~\ref{sec:intens-stre-proc}, we characterised the set of
intensional stream processors as a final $\mathbb{T}_A$-residual
$\mathbb{T}_B$-comodel. In this section, we give the main result of
the paper, characterising the set of extensional stream processors
$\cat{Top}(A^\mathbb{N}, B^\mathbb{N})$ as a final
\emph{bimodel}~\cite{Freyd1966Algebra, Tall1970Representable,
  Bergman1996Co-groups} for $\mathbb{T}_A$ and $\mathbb{T}_B$.

\begin{defi}[Bimodel]
  Let $\mathbb{T}$ and $\mathbb{R}$ be theories. An
  \emph{$\mathbb{R}$-$\mathbb{T}$-bimodel} $\alg[K]$ is an
  $\mathbb{R}$-model $(\alg[K], \dbr{\thg}_{\alg[K]})$ endowed with
  $\mathbb{T}$-comodel structure $\dbr{\thg}^{\alg[K]}$ in the category
  $\cat{Mod}(\mathbb{R})$ of $\mathbb{R}$-models.
\end{defi}

The main difficulty in working with $\mathbb{R}$-$\mathbb{T}$-bimodels
is handling copowers in $\cat{Mod}(\mathbb{R})$. A simple case is that
of \emph{free} $\mathbb{R}$-models: a copower of free models is free,
and so we have canonical isomorphisms
$B \cdot \alg[R](V) \cong \alg[R](B \times V)$, which for convenience,
we will assume are in fact \emph{identities}, i.e., that the chosen
copower $B \cdot \alg[R](V)$ is $\alg[R](B \times V)$. The
$\mathbb{R}$-$\mathbb{T}$-bimodels with free underlying
$\mathbb{R}$-model are easy to identify: they correspond precisely to
$\mathbb{R}$-residual $\mathbb{T}$-comodels, where the
$\mathbb{R}$-$\mathbb{T}$-bimodel $\alg[R](\alg[S])$ corresponding to
the $\mathbb{R}$-residual $\mathbb{T}$-comodel $\alg[S]$ has
underlying $\mathbb{R}$-model $\alg[R](S)$ and co-operations
$\dbr{\sigma}^{\alg[R](\alg[S])} =
\bigl(\dbr{\sigma}^{\alg[S]}\bigr)^\dagger \colon \alg[R](S)
\rightarrow \alg[R](\abs{\sigma} \times S)$, where $(\thg)^\dagger$ is
the Kleisli extension operation of Lemma~\ref{lem:4}.

\looseness=-1
To understand what we gain by looking at bimodels with non-free
underlying model, it is helpful to think in terms of quotients by
bisimulations. If $\alg[S]$ is an $\mathbb{R}$-residual
$\mathbb{T}$-comodel, then we could define
(cf.~\cite[Definition~5.2]{Pattinson2015Sound}) a \emph{bisimulation}
on $\alg[S]$ to be an equivalence relation $\sim$ on $S$ such that
each co-operation $\dbr{\sigma}^{\alg[S]}$ sends $\sim$-related states
to $\approx$-related computations in $R(\abs{\sigma} \times S)$, where
$\approx$ is the congruence generated by $(i, s) \approx (i, s')$
whenever $s \sim s'$. The definition ensures that the residual comodel
structure descends to the quotient set $S \quot \sim$; however, this
only gives the possibility of identifying operationally equivalent
\emph{states}, and not operationally equivalent \emph{computations}
over states. The following more generous definition rectifies this.

\begin{defi}[$\mathbb{R}$-bisimulation]
  Let $\mathbb{R}$ and $\mathbb{T}$ be theories. For any
  $\mathbb{R}$-congruence $\sim$ on the free model $\alg[R](V)$ and
  any set $B$, the congruence $\sim_{B}$ on
  $\alg[R](B \times V)$ is that generated by
  \begin{equation*}
    t \sim u \text{ in } R(V) \qquad \implies \qquad
    t(\lambda s. (b,s)) \sim_{B} u(\lambda s. (b,s)) \text{ for all $b \in B$.}
  \end{equation*}
  If $\alg[S]$ is an $\mathbb{R}$-residual $\mathbb{T}$-comodel, then
  a congruence on $\alg[R](S)$ is an \emph{$\mathbb{R}$-bisimulation}
  if the co-operations
  $\dbr{\sigma}^{\alg[R](\alg[S])}  =
  \bigl(\dbr{\sigma}^{\alg[S]}\bigr)^\dagger \colon \alg[R](S) \rightarrow
  \alg[R](\abs \sigma \times S)$ of the associated bimodel send
  $\sim$-congruent terms to $\sim_{\abs{\sigma}}$-congruent terms.
\end{defi}
\begin{lem}
  \label{lem:2}
  Let $\alg[S]$ be an $\mathbb{R}$-residual $\mathbb{T}$-comodel and
  $\sim$ an $\mathbb{R}$-bisimulation on $\alg[R](S)$. There is a
  unique structure of $\mathbb{R}$-$\mathbb{T}$-bimodel on the
  quotient $\mathbb{R}$-model $\alg[K] = \alg[R](S) \quot \sim$ for
  which the the quotient map
  $q \colon \alg[R](S) \twoheadrightarrow \alg[K]$ becomes a map of
  bimodels $\alg[R](\alg[S]) \twoheadrightarrow \alg[K]$.
\end{lem}
\begin{proof}
  If $\alg[R](S) \quot \mathord{\sim} = \alg[K]$ then
  $\alg[R](B \times S) \quot \sim_B$ is a presentation of the copower
  $B \cdot \alg[K]$. So the assumption that $\sim$ is an
  $\mathbb{R}$-bisimulation ensures that each basic co-operation of
  $\alg[R](\alg[S])$ descends to a co-operation on $\alg[K]$, as to the
  left in:
  \begin{equation*}
    \cd[@-0.5em@C+0.5em]{
      \alg[R](S) \ar[r]^-{\dbr{\sigma}^{\alg[R](\alg[S])}} \ar@{->>}[d]_-{q} &
      \alg[R](\abs \sigma \times S) \ar@{->>}[d]^-{q_{\abs \sigma}} & & 
      \alg[R](S) \ar[r]^-{\dbr{t}^{\alg[R](\alg[S])}} \ar@{->>}[d]_-{q} &
      \alg[R](V \times S) \ar@{->>}[d]^-{q_V} \\
      \alg[K] \ar@{-->}[r]^-{\dbr{\sigma}^{\alg[K]}} & \abs{\sigma} \cdot \alg[K] & & 
      \alg[K] \ar@{-->}[r]^-{\dbr{t}^{\alg[K]}} & V \cdot \alg[K]\rlap{ .}
    }
  \end{equation*}
  Since $\alg[R](\abs{\sigma} \times S)$ is the copower
  $\abs \sigma \cdot \alg[R](S)$, and the quotient map
  $q_{\abs \sigma}$ is the copower $\abs \sigma \cdot q$, it follows
  that the \emph{derived} co-operations of $\alg[R](\alg[S])$ descend
  to the corresponding \emph{derived} co-operations on $\alg[K]$, as
  to the right above; whence the satisfaction of the
  $\mathbb{T}$-comodel equations for $\alg[R](\alg[S])$ implies the
  corresponding satisfaction for $\alg[K]$. So $\alg[K]$ is an
  $\mathbb{R}$-$\mathbb{T}$-bimodel, and clearly this is the
  \emph{unique} bimodel structure making $q$ into a bimodel
  homomorphism.
\end{proof}

We can use this construction to explain how the passage from the final
$\mathbb{T}_A$-residual $\mathbb{T}_B$-comodel to the final
$\mathbb{T}_A$-$\mathbb{T}_B$-bimodel will collapse the intensionality
we saw in Example~\ref{ex:3}.

\begin{exa}
  Consider the two $\mathbb{T}_A$-residual $\mathbb{T}_B$-comodels
  $\alg[S]_1$ and $\alg[S]_2$ of Example~\ref{ex:3}. While there is
  clearly no scope for quotienting by a bisimulation on the set of
  states $\{\ast\}$, we \emph{can} non-trivially quotient each by a
  $\mathbb{T}_A$-bisimulation on $\alg[T]_A(\ast)$: namely the
  $\mathbb{T}_A$-congruence on $\alg[T]_A(\ast)$ generated by
  $\ast \sim \mathsf{read}(\lambda a.\, \ast)$. It is easy to see that
  this is a $\mathbb{T}_A$-bisimulation for both $\alg[S]_1$ and
  $\alg[S]_2$, and so we obtain quotient
  $\mathbb{T}_A$-$\mathbb{T}_B$-bimodels
  $\alg[T]_A(\alg[S]_1) \quot \sim$ and
  $\alg[T]_A(\alg[S]_2) \quot \sim$. In fact, these are visibly the
  \emph{same} bimodel $\alg[K]$, with underlying $\mathbb{T}_A$-model
  the final model $\{\ast\}$, and with $\mathbb{T}_B$-comodel
  structure
  $\dbr{\mathsf{read}}^{\alg[K]} \colon \alg[K] \rightarrow B \cdot
  \alg[K]$ given by the $b$th coproduct injection. So we have a cospan
  of bimodels $\alg[T]_A(\alg[S]_1) \twoheadrightarrow \alg[K]
  \twoheadleftarrow \alg[T]_A(\alg[S]_2)$, which in particular implies
  that the states $\ast$ of $\alg[T]_A(\alg[S]_1)$ and
  $\alg[T]_A(\alg[S]_2)$ must be identified in a \emph{final}
  $\mathbb{T}_A$-$\mathbb{T}_B$-bimodel.
\end{exa}

This example provides supporting evidence for the main theorem we
shall now prove: that the set $\cat{Top}(A^\mathbb{N}, B^\mathbb{N})$
of extensional stream processors is a final
$\mathbb{T}_A$-$\mathbb{T}_B$-bimodel. Before giving this, let us
mention some related results:

\begin{prop}
  For any set $B$, the final $\mathbb{P}_f^+$-$\mathbb{T}_B$-bimodel
  is given by the set of topologically closed subsets of
  $B^\mathbb{N}$, with the $\mathbb{P}_f^+$-model structure given by
  binary union, and with $\mathbb{T}_B$-comodel structure given as
  in~\cite[Proposition~71]{Garner2018Hypernormalisation}.
\end{prop}
\begin{prop}
  For any set $B$, the final $\mathbb{D}$-$\mathbb{T}_B$-bimodel is
  given by the set of probability distributions on $B^\mathbb{N}$,
  with the $\mathbb{D}$-model structure given by the usual convex
  combination of probability distributions, and with
  $\mathbb{T}_B$-comodel structure given as
  in~\cite[Proposition~79]{Garner2018Hypernormalisation}.
\end{prop}
In both of these examples, the elements of the final bimodel can be
seen as \emph{traces} for states of the corresponding residual
comodels (which, we recall, are non-terminating labelled transition
systems and finitely branching probabilistic generative systems
respectively.) It is thus consistent for us to think of the final
$\mathbb{T}_A$-$\mathbb{T}_B$-bimodel
as providing an ``object of
traces'' for $A$-$B$-stream processors.

To show that this final bimodel can be identified with
$\cat{Top}(A^\mathbb{N}, B^\mathbb{N})$, we will first construct an
adjunction as to the left in
\begin{equation}\label{eq:9}
  \cd[@C+2em]{
    {\cat{Mod}(\mathbb{T}_A)} \ar@<-4.5pt>[r]_-{(\thg) \otimes \alg[A^\mathbb{N}]} \ar@{<-}@<4.5pt>[r]^-{\cat{Top}(\alg[A^\mathbb{N}], \thg)} \ar@{}[r]|-{\top} &
    {\cat{Top}} 
  } \qquad 
    \cd[@C+2em]{
    {\cat{Comod}(\mathbb{T}_B, \cat{Mod}(\mathbb{T}_A))} \ar@<-4.5pt>[r]_-{(\thg) \otimes \alg[A^\mathbb{N}]} \ar@{<-}@<4.5pt>[r]^-{\cat{Top}(\alg[A^\mathbb{N}], \thg)} \ar@{}[r]|-{\top} &
    {\cat{Comod}(\mathbb{T}_B, \cat{Top})} 
  }
\end{equation}
We then show that
\emph{both} directions of this adjunction preserve coproducts, so in
particular copowers; it will then follow that the adjunction to the
left lifts to one as to the right on $\mathbb{T}_B$-comodels. The
right adjoint of this lifted adjunction, like any right adjoint,
will preserve terminal objects, and so must send the final topological
$\mathbb{T}_B$-comodel $\alg[B^\mathbb{N}]$ to a final
$\mathbb{T}_A$-$\mathbb{T}_B$-bimodel, with underlying set
$\cat{Top}(A^\mathbb{N}, B^\mathbb{N})$.

To construct the adjunction to the left in~\eqref{eq:9} we apply a
standard result of category-theoretic universal
algebra~(cf.~\cite[Theorem~2]{Freyd1966Algebra}). For 
self-containedness we give a full proof.

\begin{prop}
  \label{prop:1}
  Let $\C$ be a category with copowers and $\alg[S]$ a
  $\mathbb{T}$-comodel in $\C$. For any object $C \in \C$, the hom-set
  $\C(S, C)$ bears a structure of $\mathbb{T}$-model $\C(\alg[S], C)$ with operations
  \begin{equation}\label{eq:16}
    \dbr{\sigma}_{\C(\alg[S], C)}(\lambda i.\, S \xrightarrow{f_i} C) = S \xrightarrow{\dbr{\sigma}^{\alg[S]}} \abs{\sigma} \cdot S \xrightarrow{\langle f_i\rangle_{i \in \abs{\sigma}}} C
  \end{equation}
  where $\langle f_i \rangle_{i \in \abs{\sigma}}$ is the copairing of
  the $f_i$'s. As $C$ varies, this assignment underlies a functor
  $\C(\alg[S], \thg) \colon \C \rightarrow \cat{Mod}(\mathbb{T})$.
  If $\C$ is cocomplete, this functor has a left adjoint
  $(\thg) \otimes \alg[S] \colon \cat{Mod}(\mathbb{T}) \rightarrow
  \C$.
\end{prop}
\begin{proof}
  For any $C \in \C$, the hom-functor
  $\C(\thg, C) \colon \C^\mathrm{op} \rightarrow \cat{Set}$ sends
  copowers in $\C$ to powers in $\cat{Set}$, and so sends
  $\mathbb{T}$-comodels in $\C$ to $\mathbb{T}$-models in $\cat{Set}$.
  In particular, the $\mathbb{T}$-model induced by $\alg[S] \in
  \cat{Comod}(\mathbb{T}, \C)$ is $\C(\alg[S], C)$ with the operations
  defined above. The functoriality of this assignment is clear, so it
  remains to exhibit the desired adjoint when $\C$ is cocomplete.

  To this end, note that a $\mathbb{T}$-model homomorphism $\alpha \colon \alg[X]
  \rightarrow \C(\alg[S], C)$ is equally a function $\alpha \colon X
  \rightarrow \C(S,C)$ such that, for all basic
  $\mathbb{T}$-operations $\sigma$ and all $\vec x \in X^{\abs
    \sigma}$, we have
  \begin{equation*}
    S \xrightarrow{\alpha(\dbr{\sigma}_{\alg}(\vec x))} C \qquad  =\qquad  S \xrightarrow{\dbr{\sigma}^{\alg[S]}} \abs{\sigma} \cdot S \xrightarrow{\langle \alpha(x_i) \rangle_{i \in \abs{\sigma}}} C\rlap{ .}
  \end{equation*}
  Transposing under $(\thg) \cdot S \dashv \C(S, \thg)
  \colon \C \rightarrow \cat{Set}$, this is equally to give a
  map $\bar \alpha \colon X \cdot S \rightarrow C$ in $\C$ such that,
  for each basic $\mathbb{T}$-operation $\sigma$, postcomposition with $\bar
  \alpha$ equalises the two maps
  \begin{equation*}
    X^{\abs{\sigma}} \cdot S \xrightarrow{\dbr{\sigma}_{\alg} \cdot C} X \cdot S \qquad \quad
    X^{\abs{\sigma}} \cdot S \xrightarrow{X^{\abs \sigma} \cdot \dbr{\sigma}^{\alg}} X^{\abs \sigma} \cdot (\abs \sigma \cdot S) \cong
    (X^{\abs \sigma} \times \abs \sigma) \cdot S \xrightarrow{\mathrm{ev} \cdot S} X \cdot S\text{ .}
  \end{equation*}
  Thus, defining $\alg \otimes \alg[S]$ to be the joint coequaliser of
  these parallel pairs as $\sigma$ varies across the basic
  $\mathbb{T}$-operations, we have bijections
  $\C(\alg \otimes \alg[S], C) \cong \cat{Mod}(\mathbb{T})(\alg,
  \C(\alg[S], C))$ natural in $C \in \C$, so that $\alg \otimes
  \alg[S]$ is the value at $\alg$
  of the desired left adjoint $(\thg) \otimes \alg[S]$.
\end{proof}

\begin{rem}
  \label{rk:2}
  Again, the final part of this result expresses an abstract fact:
  $\cat{Mod}(\mathbb{T})$ is the \emph{free cocomplete category
    containing a comodel of $\mathbb{T}$}. 
\end{rem}

In particular, we may apply the preceding result when $\C$ is the
cocomplete category $\cat{Top}$ and $\alg[S]$ is the final topological
$\mathbb{T}_A$-comodel $\alg[A^\mathbb{N}]$ to obtain an adjunction as
to the left in~\eqref{eq:9}. We now show that both directions of this
adjunction preserve coproducts, and so in particular copowers. Since
left adjoints always preserve colimits, there is only work to
do for the right adjoint
$\cat{Top}(\alg[A^\mathbb{N}], \thg) \colon \cat{Top} \rightarrow
\cat{Mod}(\mathbb{T}_A)$. First we spell out that, on objects, this
functor acts by taking a space $C$ to the set of continuous functions
$\cat{Top}(A^\mathbb{N}, C)$, under the $A$-ary magma structure
$\mathsf{split}$ that takes a family $(f_a : a \in A)$ of 
functions to the function $\mathsf{split}(\vec f)$ with
\begin{equation}\label{eq:10}
  \mathsf{split}(\vec f)(\vec a) = f_{a_0}(\partial \vec a)\rlap{ .}
\end{equation}
In other words, $\mathsf{split}(\vec f)$ consumes the first token
$a_0$ of its input and then continues as $f_{a_0}$ on the rest of its
input; note that $\mathsf{split}$ is in fact invertible, with
inverse given by the function
$\mathsf{split}^{-1}(f) = (f(a\thg) : a \in A)$. This describes the
action of 
$\cat{Top}(\alg[A^\mathbb{N}], \thg) \colon \cat{Top} \rightarrow
\cat{Mod}(\mathbb{T}_A)$ on objects; on morphisms, it simply acts by postcomposition. 

The following result is the main piece of serious work needed to
complete our result; it refines the topological arguments described
in~\cite[Theorem~2.1]{Hancock2009Representations}, and used there to
construct the reification function for intensional stream processors.

\begin{prop}
  \label{prop:4}
  The functor $\cat{Top}(\alg[A^\mathbb{N}], \thg) \colon \cat{Top} \rightarrow
  \cat{Mod}(\mathbb{T}_A)$ preserves coproducts.
\end{prop}
\begin{proof}
    Given spaces $(X_i : i \in I)$, we have the coproduct
  injections $\iota_i \colon X_i \rightarrow \Sigma_i X_i$ in
  $\cat{Top}$, and must show that the
  family of postcomposition maps
  \begin{equation}
    \label{eq:6}
      \bigl(\,\iota_i \circ (\thg) \colon \cat{Top}(\alg[A^\mathbb{N}], X_i) \rightarrow \cat{Top}(\alg[A^\mathbb{N}], \Sigma_i X_i)\,\bigr)_{i \in I}
  \end{equation}
  constitute a coproduct cocone in $\cat{Mod}(\mathbb{T}_A)$. We first show:
  \begin{lem}
    \label{lem:3}
    The maps~\eqref{eq:6} are jointly epimorphic in
    $\cat{Mod}(\mathbb{T}_A)$.
  \end{lem}
  \begin{proof}
    We show that the sub-$A$-ary magma
    $\alg[M] \subseteq \cat{Top}(\alg[A^\mathbb{N}], \Sigma_i X_i)$
    generated by the image of the maps~\eqref{eq:6} is all of
    $\cat{Top}(\alg[A^\mathbb{N}], \Sigma_i X_i)$. So suppose not;
    then there exists some continuous
    $f \colon A^\mathbb{N} \rightarrow \Sigma_i X_i$ with
    $f \notin M$. Since we have 
    $f = \mathsf{split}(\mathsf{split}^{-1}(f)) = \mathsf{split}(\lambda a.\, f(a\thg))$, we can find $a_0 \in A$ with
    $f(a_0\thg) \notin M$. Now repeating the argument with
    $f(a_0\thg)$, we can find $a_1 \in A$ with
    $f(a_0 a_1 \thg) \notin M$; and continuing in this fashion, making
    countably many dependent choices, we find some
    $\vec a \in A^\mathbb{N}$ such that for all $n$,
    the continuous function
    $f(a_0 a_1 \dots a_n \thg) \colon A^\mathbb{N} \rightarrow \Sigma_i
    X_i$ is not in $M$. In particular, none of these functions factor
    through any $X_i$; but as $f(\vec a) \in X_i$ for some $i$,
    this means there is \emph{no} open neighbourhood of $\vec a$ which
    is mapped by $f$ into the open neigbourhood $X_i$ of $f(\vec a)$,
    contradicting the continuity of $f$.
  \end{proof}
  Thus, to complete the proof, we need only show that, for a cocone
  $( p_i \colon \cat{Top}(\alg[A^\mathbb{N}], X_i) \rightarrow \alg[Y])_{i \in I}$
  in $\cat{Mod}(\mathbb{T}_A)$, there exists
  \emph{some} map $p \colon \cat{Top}(\alg[A^\mathbb{N}], \Sigma_i X_i) \rightarrow \alg[Y]$
  with $p \circ \cat{Top}(\alg[A^\mathbb{N}],\iota_i) = p_i$ for each $i$. To this end, consider
  the diagram of $A$-ary magmas
  \begin{equation*}
    \cd[@!C@C-3em@-1em]{
      & {\alg[N]} \ar@{->>}[dl]_-{\varepsilon} \ar[dr]^-{\tilde p} \\
      {\cat{Top}(\alg[A^\mathbb{N}], \Sigma_i X_i)} \ar@{-->}[rr]^{p} & &
      {\alg[Y]}
    }
  \end{equation*}
  where $\alg[N] = (N, \nu)$ is the free $A$-ary magma generated by symbols
  $[f,i]$ for $i \in I$ and
  $f \in \cat{Top}(\alg[A^\mathbb{N}], X_i)$, where $\varepsilon$
  sends $[f,i]$ to $\iota_i f$ and where $\tilde p$ sends $[f,i]$ to
  $p_i(f)$. It suffices to exhibit a factorisation $p$ of $\tilde p$
  through $\varepsilon$ as displayed. Now by the lemma above,
  $\varepsilon$ is epimorphic, and so the coequaliser of its
  kernel-congruence; so to obtain such a factorisation, it suffices to
  show that if $x,y \in N$ satisfy
  $\varepsilon(x) = \varepsilon(y)$, then they satisfy
  $\tilde p(x) = \tilde p(y)$. We do so by induction on the total
  number of magma operations $\nu$ in $x$ and $y$:
  \begin{itemize}
  \item If $x = [f,i]$ and $y = [g,j]$ then
    $\varepsilon(x) = \varepsilon(y)$ says that
    $\iota_i f = \iota_j g$, which is possible only if $i = j$ and
    $f = g$. So $x = y$ and so certainly $\tilde p(x) = \tilde p(y)$.
  \item If $x = [f,i]$ and $y = \nu(\lambda a.\, y_a)$ then on taking
    $x_a = [f(a\thg), i]$ for each $a$, we get from $\varepsilon(x)
    = \varepsilon(y)$ that
    \begin{equation*}
      \mathsf{split}(\lambda a.\, \varepsilon(x_a)) =
      \mathsf{split}(\lambda a.\, \iota_i f(a\thg)) =
      \iota_i f = \varepsilon(x) = \varepsilon(y) = \varepsilon(\nu(\lambda a.\, y_a)) = \mathsf{split}(\lambda a.\, \varepsilon(y_a))
    \end{equation*}
    which, since $\mathsf{split}$ is invertible, implies that
    $\varepsilon(x_a) = \varepsilon(y_a)$ for each $a \in A$. By
    induction, we have $\tilde p(x_a) = \tilde p(y_a)$ for each $a$,
    and so we have the desired equality:
    \begin{equation*}
      \tilde p(x) = p_i(f) = \mathsf{split}(\lambda a.\, p_i(f(a\thg))) =
      \mathsf{split}(\lambda a.\, \tilde p(x_a)) =
      \mathsf{split}(\lambda a.\, \tilde p(y_a)) = 
      \tilde p(\nu(\lambda a.\, y_a) = \tilde p(y)\rlap{ .}
    \end{equation*}

  \item The case where $x = \nu(\lambda a.\, x_a)$ and $y = [g,j]$ is dual.
    
  \item Finally, if $x = \nu(\lambda a.\, x_a)$ and $y = \nu(\lambda
    a.\, y_a)$, then from $\varepsilon(x) = \varepsilon(y)$ we get
    \begin{equation*}
      \mathsf{split}(\lambda a.\, \varepsilon(x_a)) = \varepsilon(\nu(\lambda a.\, x_a)) = \varepsilon(x) = \varepsilon(y) = \varepsilon(\nu(\lambda a.\, y_a)) =
      \mathsf{split}(\lambda a.\, \varepsilon(y_a))
    \end{equation*}
    and so by invertibility of $\mathsf{split}$ that
    $\varepsilon(x_a) = \varepsilon(y_a)$ for all $a$. By induction,
    $\tilde p(x_a) = \tilde p(y_a)$ for all $a$,
    and so the desired equality
    \begin{equation*}
      \tilde p(x) = \tilde p(\nu(\lambda a.\, x_a) = \mathsf{split}(\lambda a.\, \tilde p(x_a)) =
      \mathsf{split}(\lambda a.\, \tilde p(y_a)) = 
      \tilde p(\nu(\lambda a.\, y_a) = \tilde p(y)\rlap{ .} \qedhere
    \end{equation*}
  \end{itemize}
\end{proof}

Using this result, we can conclude the argument as explained above.
Since both adjoints to the left of~\eqref{eq:9} preserve
coproducts, the adjunction lifts to an adjunction between
categories of $\mathbb{T}_B$-comodels as to the right. In particular,
the lifted right adjoint sends the final topological
$\mathbb{T}_B$-comodel to a final
$\mathbb{T}_A$-$\mathbb{T}_B$-bimodel, so giving our main theorem:

\begin{thm}
  \label{thm:2}
  For any sets $A$ and $B$, the final $\mathbb{T}_A$-$\mathbb{T}_B$-bimodel
  $\alg[E_{AB}]$ is given by the set of continuous functions
  $\cat{Top}(A^\mathbb{N}, B^\mathbb{N})$ with
  the $\mathbb{T}_A$-model structure of~\eqref{eq:10}, and with the
  $\mathbb{T}_B$-comodel structure map
  \begin{equation}\label{eq:12}
    \cat{Top}(\alg[A^\mathbb{N}], B^\mathbb{N})
    \xrightarrow{(g,n) \circ (\thg)} \cat{Top}(\alg[A^\mathbb{N}],
    B \cdot B^{\mathbb{N}})
    \xrightarrow{\cong}
    B \cdot \cat{Top}(\alg[A^\mathbb{N}],
    B^{\mathbb{N}})\rlap{ ,}
  \end{equation}
  whose first part is postcomposition with~\eqref{eq:3} and whose
  second part is the canonical isomorphism coming from the fact that
  $\cat{Top}(\alg[A^\mathbb{N}], \thg) \colon \cat{Top} \rightarrow
  \cat{Mod}(\mathbb{T}_A)$ preserves coproducts.
\end{thm}

We now describe~\eqref{eq:12} more concretely, but first
we describe copowers in
$\cat{Mod}(\mathbb{T}_A)$.

\begin{lem}
  \label{lem:5}
  For any $\mathbb{T}_A$-model $\alg = (X, \xi)$ and set $B$, the
  copower $B \cdot \alg$ may be found as either: (i) the quotient of $\alg[T]_A(B \times
  X)$ by the congruence which identifies
  \begin{equation}\label{eq:13}
    \cd[@-1em]{
      (b, x_a) \ar@{-}[dr]^-{} & \cdots & (b, x_{a'}) \ar@{-}[dl]_-{} \\
      & \bullet \ar@{-}[d]_-{} \\
      & {}
    } \qquad \qquad \sim \qquad \qquad 
    \cd[@-1.3em]{\phantom{(b,f_a)}
      \\(b, \xi(\lambda a.\, x_a))\rlap{ ;} \ar@{-}[d]_-{} \\ {}
    } 
  \end{equation} or: (ii) the subset of
  $T_A(B \times X)$ on those $A$-ary branching trees where no
  non-trivial subtree has all its leaves labelled by the same element
  of $B$, with the $\mathbb{T}_A$-model structure map $\upsilon$ being that of $T_A(B \times
  X)$ except that $\upsilon(\lambda a.\, (b, x_a)) = (b, \xi(\lambda a.\,
  x_a))$.
\end{lem}
\begin{proof}
  (i) is the presentation $\alg[T]_A(B \times X) \quot \sim_B$ from
  Lemma~\ref{lem:2} when $\sim$ is the congruence associated to the
  quotient
  $\mathrm{id}^\dagger \colon \alg[T]_A(X) \twoheadrightarrow \alg$. As for
  (ii), these elements are the normal forms for the strongly
  normalising rewrite system obtained by applying~\eqref{eq:13} from
  left to right.
\end{proof}

Via presentation (i), we may thus describe~\eqref{eq:12} by
associating to each $f \in \cat{Top}(A^\mathbb{N}, B^\mathbb{N})$ a
suitable tree in
$T_A(B \times \cat{Top}(A^\mathbb{N}, B^\mathbb{N}))$. For this, we
use the identification of $B^\mathbb{N}$ with $B \cdot B^\mathbb{N}$
via $\vec b \mapsto (b_0, \partial \vec b)$, together with
Lemma~\ref{lem:3}, to see that $f$ lies in the closure under the
$A$-ary magma operation $\mathsf{split}$ on
$\cat{Top}(A^\mathbb{N}, B^\mathbb{N})$ of the set of those
$g \colon A^\mathbb{N} \rightarrow B^\mathbb{N}$ for which
$g(\vec a)_0$ is constant. (This expresses algebraically the fact
that, for each $\vec a \in A^\mathbb{N}$, there is some finite initial
segment $a_0 \dots a_k$ of $\vec a$ such that 
$f(\vec a')_0 = f(\vec a)_0$ whenever
$a_0 \dots a_k = a'_0 \dots a'_k$.)

Choosing any such presentation of $f$ gives a well-founded $A$-ary
tree (encoding the applications of $\mathsf{split}$) with leaves
labelled by functions $g \colon A^\mathbb{N} \rightarrow B^\mathbb{N}$
with $g(\vec a)_0$ constant. Each such $g$ is equally specified by the
constant $b = g(\vec a)_0$, and the function
$h = \partial \circ g \colon A^\mathbb{N} \rightarrow B^\mathbb{N}$,
so that our leaf labels are equally elements in
$B \times \cat{Top}(A^\mathbb{N}, B^\mathbb{N})$: so altogether we
have an element of
$T_A(B \times \cat{Top}(A^\mathbb{N}, B^\mathbb{N}))$. Note that
choosing a different presentation of $f$ would yield a different
element of $T_A(B \times \cat{Top}(A^\mathbb{N}, B^\mathbb{N}))$;
however, our theory ensures that these elements are congruent
under~\eqref{eq:13}, so yielding a well-defined element of
$B \cdot \cat{Top}(\alg[A^\mathbb{N}], B^\mathbb{N})$.

\section{Comparing intensional and extensional stream processors}
\label{sec:comp-intens-extens}

To conclude the paper, we examine the unique maps from an arbitrary
$\mathbb{T}_A$-$\mathbb{T}_B$-bimodel to the final one, showing that
these act as expected via the trace function of
Definition~\ref{def:5}; and, finally, we give a comodel-theoretic
explanation of ``normalisation-by-trace-evaluation'' for intensional
stream processors.

We begin with a small refinement of Proposition~\ref{prop:1}.

\begin{propC}[{\cite[Theorem~4.4]{Plotkin2008Tensors}}]
  \label{prop:5}
  Let $\C$ be a cocomplete category and $\alg[S]$ a
  $\mathbb{T}$-comodel in $\C$. The functor $(\thg) \otimes \alg[S]$
  of Proposition~\ref{prop:1} may be chosen to render
  commutative the following diagram, whose top edge is as
  in~\eqref{eq:2}, and whose left edge is as in Definition~\ref{def:7}:
  \begin{equation}\label{eq:14}
    \cd[@-0.7em]{
      \cat{Kl}(\mathbb{T}) \ar[r]^-{(\thg) \cdot \alg[S]} \ar[d]_-{I_\mathbb{T}} & \C
      \rlap{ .}\\
      \cat{Mod}(\mathbb{T}) \ar[ur]_-{(\thg) \otimes \alg[S]}
    }
  \end{equation}
\end{propC}

\begin{proof}
    For a free
  $\mathbb{T}$-model $\alg[T](V)$, we have natural bijections
  $\cat{Mod}(\mathbb{T})(\alg[T](V), \C(\alg[S], C)) \cong
  \cat{Set}(V, \C(S, C)) \cong \C(V \cdot S, C)$, and so we may take
  $\alg[T](V) \otimes \alg[S] = V \cdot S$. This makes~\eqref{eq:14}
  commute on objects. On morphisms, given 
  $\theta^\dagger \colon \alg[T](V) \rightarrow \alg[T](W)$ in
  $\cat{Mod}(\mathbb{T})$, its image 
  $\theta^\dagger \otimes \alg[S] \colon V \cdot S \rightarrow W \cdot
  S$ under $(\thg) \otimes \alg[S]$ is, by adjointness, the unique map making:
  \begin{equation*}
    \cd[@-0.7em]{
      \C(W \cdot S, C) \ar[r]^-{\cong} \ar[d]_-{(\thg) \circ (\theta^\dagger \otimes \alg[S])} &
      \cat{Set}(W, \C(S,C)) \ar[r]^-{\cong}
      \ar@{.>}[d] &
      \cat{Mod}(\mathbb{T})(\alg[T]W, \C(\alg[S], C)) \ar[d]^-{(\thg) \circ \theta^+}\\
      \C(V \cdot S, C) \ar[r]^-{\cong} &
      \cat{Set}(V, \C(S,C)) \ar[r]^-{\cong} & \cat{Mod}(\mathbb{T})(\alg[T]V, \C(\alg[S], C))
    }
  \end{equation*}
  commute for all $C \in \C$. Now, the unique dotted map making the
  right square commute is, by the freeness of $\alg[T](V)$, the
  function
  \begin{equation*}
    (f_w \in \C(S,C) : w \in W)  \qquad \mapsto \qquad
    (\dbr{\theta(v)}_{\C(\alg[S],C)}(\vec f) : v \in V)\rlap{ .}
  \end{equation*}
  But by induction on~\eqref{eq:16}, we have
  $\dbr{\theta(v)}_{\C(\alg[S], C)}(\vec f) = S
  \xrightarrow{\dbr{\theta(v)}^{\alg[S]}} W \cdot S
  \xrightarrow{\langle f_i\rangle_{i \in \abs{\sigma}}} C$; whence we
  must have
  $\theta^\dagger \otimes \alg[S] = \langle \dbr{\theta(v)}^{\alg[S]}
  \rangle_{v \in V} = \theta \cdot \alg[S]$ as desired.
\end{proof}

We now characterise the unique maps to $\alg[E_{AB}]$ from bimodels
induced by residual comodels.

\begin{prop}
  \label{prop:6}
  Let $\alg[S]$ be a $\mathbb{T}_A$-residual $\mathbb{T}_B$-comodel.
  The unique bimodel map from the associated bimodel
  $f \colon \alg[T]_A(\alg[S]) \rightarrow \alg[E_{AB}]$ is $\mathsf{tr}^\dagger$, the
  homomorphic extension of the trace function
  $\mathsf{tr} \colon S \rightarrow \cat{Top}(A^\mathbb{N},
  B^\mathbb{N})$ of Definition~\ref{def:5}.
\end{prop}
\begin{proof}
  Since
  $(\thg) \otimes \alg[A^\mathbb{N}] \colon \cat{Mod}(\mathbb{T}_A)
  \rightarrow \cat{Top}$ restricts back along $I_{\mathbb{T}_A}$ to
  $(\thg) \cdot \alg[A^\mathbb{N}] \colon \cat{Kl}(\mathbb{T}_A)
  \rightarrow \cat{Top}$, its lifting to a functor on
  $\mathbb{T}_B$-comodels must, by Remark~\ref{rk:4}, restrict along
  $I_{\mathbb{T}_A}$ to the tensor product of Definition~\ref{def:2}.
  So the unique $\mathbb{T}_B$-comodel map
  $\alg[T]_A(\alg[S]) \otimes \alg[A^\mathbb{N}] \rightarrow
  \alg[B^\mathbb{N}]$ must be the unique map
  $\alg[S] \cdot \alg[A^\mathbb{N}] \rightarrow \alg[B^\mathbb{N}]$ of
  Definition~\ref{def:5}. By the proof of Proposition~\ref{prop:5},
  transposing this latter map to a bimodel map
  $\alg[T]_A(\alg[S]) \rightarrow \alg[E_{AB}]$ is achieved by first
  currying---which yields the trace function $\mathsf{tr} \colon S
  \rightarrow \cat{Top}(A^\mathbb{N}, B^\mathbb{N})$---and then extending homomorphically.
\end{proof}

We now do the same for the unique maps to $\alg[E_{AB}]$ from
\emph{arbitrary} $\mathbb{T}_A$-$\mathbb{T}_B$-bimodels. To do so, we
show that every such bimodel arises in a canonical way from the
construction of Lemma~\ref{lem:2}. Note that this is \emph{not} true
for bimodels over arbitrary theories; it relies on a special property
of the theory $\mathbb{T}_A$, namely that it admits \emph{abstract
  hypernormalisation} in the sense
of~\cite{Garner2018Hypernormalisation}. See section 7 of
\emph{op.~cit.} for a more detailed explanation of this phenomenon.
\begin{lem}
  \label{lem:6}
  Let $\alg[K]$ be a $\mathbb{T}_A$-$\mathbb{T}_B$-bimodel. The composite
  \begin{equation*}
    \gamma = K \xrightarrow{\dbr{\mathsf{read}}^{\alg[K]}} B \cdot K \xrightarrow{\ \ \subseteq \ } T_A(B \times K)
  \end{equation*}
  where we take $B \cdot \alg[K] \subseteq T_{A}(B \times K)$ as in
  Lemma~\ref{lem:5}(ii), endows $K$ with the structure of a
  $\mathbb{T}_A$-residual $\mathbb{T}_B$-comodel $\check{\alg[K]}$.
  The congruence on $\alg[T]_A(K)$ generating the $\mathbb{T}_A$-model
  quotient map
  $\mathrm{id}_K^\dagger \colon \alg[T]_A(K) \twoheadrightarrow
  \alg[K]$ is a $\mathbb{T}_A$-bisimulation for $\check{\alg[K]}$ and
  the quotient bimodel is precisely $\alg[K]$.
\end{lem}
\begin{proof}
  Only the final sentence requires any verification; it will follow if
  we can show that the square of $\mathbb{T}_A$-model maps to the left
  below is commutative:
  \begin{equation*}
    \cd[@C+0.5em]{
      \alg[T]_A(K) \ar[r]^-{\gamma^\dagger} \ar@{->>}[d]_-{\mathrm{id}_K^\dagger} &
      \alg[T]_A(B \times K) \ar@{->>}[d]^-{B \cdot \mathrm{id}_K^\dagger} \\
      \alg[K] \ar[r]^-{\dbr{\mathsf{read}}^{\alg[K]}} & B \cdot \alg[K]
    } \qquad \qquad
    \cd[@C+0.5em]{
      K \ar[r]^-{\dbr{\mathsf{read}}^{\alg[K]}} \ar@{->>}[d]_-{\mathrm{id}} &
      B \cdot K \ar[r]^-{\iota} &
      T_A(B \times K) \ar@{->>}[d]^-{B \cdot \mathrm{id}_K^\dagger} \\
      K \ar[rr]^-{\dbr{\mathsf{read}}^{\alg[K]}} & & B \cdot K
    }
  \end{equation*}
  which by freeness will happen just when the diagram to the right
  also commutes. But the map $B \cdot \mathrm{id}_K^\dagger$ therein
  is the quotient map by the congruence of~\eqref{eq:13}, of which
  $\iota$ must be a section since it selects a family of
  equivalence-class representatives.
\end{proof}

If $\alg[K]$ is a bimodel, then $\check{\alg[K]}$ 
is the \emph{maximally lazy} realisation of $\alg[K]$ as a residual comodel,
wherein the program associated to each state $k \in \alg[K]$ reads
the absolute minimum number of input $A$-tokens required to
determine the next output $B$-token, with all subsequent reading
from $A$ handed off (via the $\mathbb{T}_A$-model structure on
$\alg[K]$) to the continuation state.

\begin{prop}
  \label{prop:7}
  Let $\alg[K]$ be a $\mathbb{T}_A$-$\mathbb{T}_B$-bimodel. The image
  of $k \in \alg[K]$ under the unique bimodel map $\alg[K] \rightarrow
  \alg[E_{AB}]$ is the continuous function $\mathsf{tr}^{\check{\alg[K]}}(k) \colon
  A^\mathbb{N} \rightarrow B^\mathbb{N}$.
\end{prop}
\begin{proof}
  By Lemma~\ref{lem:6} we have a quotient map of bimodels
  $\mathrm{id}_K^\dagger \colon \alg[T]_A(\check{\alg[K]})
  \twoheadrightarrow \alg[K]$ which of necessity fits into a commuting
  triangle
  \begin{equation*}
    \cd[@!C@-1em@C-0.8em]{
      \alg[T]_A(\check{\alg[K]}) \ar[dr]_-{!} \ar[rr]^-{\mathrm{id}_K^\dagger} & &
      \alg[K]\rlap{ .} \ar[dl]^-{!} \\ & \alg[E_{AB}]
    }
  \end{equation*}
  The left edge of this triangle is by Proposition~\ref{prop:6} the
  homomorphic extension of
  $\mathsf{tr}{\check{\alg[K]}} \colon K \rightarrow
  \cat{Top}(A^\mathbb{N}, B^\mathbb{N})$.
  Thus, tracing the element $k \in K \subseteq T_A(K)$ around the two
  sides of this triangle yields the result.
\end{proof}

Finally, we give use the above results to give a comodel-theoretic
reconstruction of the \emph{reification} of each continuous function
on streams by an intensional stream processor; this is the function
$rep_\infty$ of~\cite{Hancock2009Representations}.

\begin{defi}
  The function $\mathsf{reify} \colon E_{AB} \rightarrow I_{AB}$ is
  the underlying map of the unique residual comodel map
  $\alg[\check{E}_{AB}] \rightarrow \alg[I_{AB}]$.
\end{defi}

As the notation suggests, we have:
\begin{prop}
  \label{prop:8}
  $\mathsf{reflect} \circ \mathsf{reify} = \mathrm{id}_{E_{AB}}$.
\end{prop}
\begin{proof}
  By Proposition~\ref{prop:6}, the unique
  $\mathbb{T}_A$-$\mathbb{T}_B$-bimodel map
  $\alg[T]_A(\alg[I_{AB}]) \rightarrow \alg[E_{AB}]$ is
  $\mathsf{reflect}^\dagger$, while by Lemma~\ref{lem:6},
  $\mathrm{id}^\dagger_{E_{AB}}$ is the unique bimodel
  $\mathbb{T}_A$-$\mathbb{T}_B$-bimodel map
  $\alg[T]_A(\alg[\check{E}_{AB}]) \rightarrow \alg[E_{AB}]$. So we
  have a (necessarily commuting) triangle of
  $\mathbb{T}_A$-$\mathbb{T}_B$-bimodel maps:
  \begin{equation*}
    \cd[@-1em]{
      \alg[T]_A(\alg[\check{E}_{AB}])
      \ar[rr]^-{\alg[T]_A(\mathsf{reify})}
      \ar[dr]_-{\mathrm{id}_{E_{AB}}^{\dagger}} & & \alg[T]_A(\alg[I_{AB}]) \ar[dl]^-{\mathsf{reflect}^\dagger}\\ & \alg[E_{AB}]
    }
  \end{equation*}
  and precomposing with $\eta \colon E_{AB} \rightarrow T_A(E_{AB})$
  yields the result.
\end{proof}

This result is proved in a more general context in~\cite[\S
7.3]{Garner2018Hypernormalisation}, and as explained there, the
composite $\mathsf{reify} \circ \mathsf{reflect}$ implements
\emph{normalisation-by-trace-evaluation}: given an intensional stream
processor, it first computes its underlying trace
$A^\mathbb{N} \rightarrow B^\mathbb{N}$, and then via the reification
function produces from this a maximally lazy intensional stream
processor realising this trace. For instance, under this procedure,
the trees $\tau_1, \tau_2 \in I_{AB}$ of Example~\ref{ex:3} will both
normalise to $\tau_1$.

\bibliography{bibliography}
\bibliographystyle{alphaurl}

\end{document}